\documentclass[10pt]{article}
\usepackage{amsmath}
\usepackage{graphicx}
\usepackage{picinpar,graphicx}
\usepackage{graphics}
\usepackage{fixltx2e}
\usepackage{geometry}                % See geometry.pdf 
\usepackage[parfill]{parskip}    % Activate to begin paragraphs with an empty line rather than an indent
\usepackage{amssymb}
\usepackage{epstopdf}
\usepackage[toc,page]{appendix}
\usepackage{lscape}
\usepackage[latin1]{inputenc}
\usepackage{amsthm}
\newtheorem{thm}{Theorem}

\newtheorem{cor}{Corollary}
\usepackage{epstopdf}
\newcommand\scalemath[2]{\scalebox{#1}{\mbox{\ensuremath{\displaystyle #2}}}}

\begin{document}
\title{The Effects of Regional Vaccination Heterogeneity on Measles Outbreaks with France as a Case Study}
\author{Elaine T. Alexander$^{1},$ Savanah D. McMahon$^{2},$ Nicholas Roberts$^{3},$ Emilio Sutti$^{4},$ \\ Daniel Burkow$^{5},$ Miles Manning$^{5},$ Kamuela E. Yong$^{6,7},$ Sergei Suslov$^{6,7}$}
\date{}
\maketitle
\begin{center}
\footnotesize $^{1}$ Department of Applied Mathematics, Arizona State University, Tempe, AZ\\
\footnotesize $^{2}$ Department of Chemistry and Biochemistry, Arizona State University, Tempe, AZ\\
\footnotesize $^{3}$ School of Mathematical and Statistical Sciences, Arizona State University, Tempe, AZ\\
\footnotesize $^{4}$ Department of Physics, National University of Salta, Salta, Argentina\\
\footnotesize $^{5}$ School of Human Evolution and Social Change, Arizona State University, Tempe, AZ\\
\footnotesize $^{6}$ School of Mathematical and Statistical Sciences, Arizona State University, Tempe, AZ\\
\footnotesize $^7$  Simon A. Levin Mathematical, Computational \& Modeling Sciences Center, Arizona State University, Tempe, AZ\\
\end{center}
\begin{abstract}
{The rubeola virus, commonly known as measles, is one of the major causes of vaccine-preventable deaths among children worldwide. This is the case despite the fact that an effective vaccine is widely available.  Even in developed countries elimination efforts have fallen short as seen by recent outbreaks in Europe, which had over 30,000 cases reported in 2010. The string of measles outbreaks in France from 2008-2011 is of particular interest due to the documented disparity in regional vaccination coverage. The impact of heterogeneous vaccine coverage on disease transmission is a broad interest and the focus of this study. A Susceptible-Exposed-Infectious-Recovered (SEIR) multi-patch epidemiological model capturing the regional differences in vaccination rates and mixing is introduced.  The mathematical analysis of a two-patch system is carried out to help our understanding of the behavior of multi-patch systems. Numerical simulations are generated to aid the study of the system's qualitative dynamics.  Data from the recent French outbreaks were used to generate parameter values and to help connect theory with application.  Our findings show that heterogeneous vaccination coverage increases controlled reproduction number compared to comparable homogeneous coverage.}
\end{abstract}

\newpage
\tableofcontents
\newpage

\section{Introduction}

	Measles is a highly contagious virus from the \emph{Morbillivirus} genus \cite{Moss2012153} that  continues to affect more than 30 million individuals worldwide \cite{Ennis}. The World Health Organization (WHO) has estimated that between the years 2000 and 2003 measles has accounted for 1 in every 25 childhood deaths, significantly higher than the ratio of deaths due to other diseases such as HIV/AIDS \cite{Weisberg2007471}. This statistic reflects the high communicability of the disease, as it travels through susceptible pockets of populations at an alarmingly fast rate. It was found that there is an approximately 85\% chance that someone who comes into direct contact with the disease, such as a susceptible household contact, will become infected \cite{rowland}. Although this disease has a very high transmission rate, there are no animal reservoirs for disease resurrection \cite{Moss2012153} nor has the virus mutated enough to alter immunogenic epitopes \cite{moss2006global}. Therefore, complete eradication of the disease is theoretically possible.
	
	The disease doesn't preferentially target a certain gender or race; the prevalence of the disease in any population thus depends primarily on socioeconomic factors, environmental conditions, and the relative vaccination coverage within the region \cite{CBO9781139053518A160}. Measles epidemics tend to occur every 2 to 5 years, during the winter and spring seasons of temperate climates \cite{CBO9781139053518A160}. It is still uncertain whether this seasonality is primarily due to the actual climate conditions or the indirect social behavior and population movement that arises from these conditions.
	
\section{Biological Overview}
\subsection{Measles Virus Characteristics}
	
	The virus is transmitted through respiratory droplets present in sneezes and coughs; it initially attacks the host's respiratory tract and from there becomes systemic \cite{deVries2012248}. The initial symptoms usually occur 8-14 days after infection and are characterized by a runny nose, a gradually increasing fever, watery eyes, a cough, drowsiness, and a loss of appetite \cite{Taber}. Following these symptoms, white lesions known as Koplik's spots appear on the inside lining of the mouth opposite the molars \cite{Taber}. These lesions preclude the characteristic measles rash, which occurs two to three days later and travels from the face to the body's extremities \cite{Taber}. The rash consists of reddish patches that are approximately 3-8mm in diameter; these patches appear gradually and last from about 3 to 7 days \cite{orenstein2004clinical}. Because the characteristic rash generally appears at the end of the communicable period, the strategy of quarantining infected individuals becomes difficult to assess. There is currently no known cure for the virus \cite{Ennis}.

\subsection{Complications Associated with Infection}
	
	Although the measles virus encodes for a haemagglutinin protein that elicits a strong immune response and grants recovered individuals immunity \cite{Moss2012153}, the serious complications associated with the virus have life-long consequences that pose major public health concerns. The measles virus has the potential to affect many organ systems throughout the body.  In particular it attacks epithelial, reticuloendothelial, and white blood cells \cite{orenstein2004clinical}. Since white blood cells are necessary for proper immune function, an infected patient can often develop severe health problems that are otherwise unassociated with the initial measles infection. These complications occur in roughly 10-30\% of measles patients and account for most of the reported fatalities \cite{semba2004measles}. Common secondary infections include bacterial ear infections, pneumonia, diarrhea, and otitis media \cite{semba2004measles}.
	
	 In rare cases, some infected individuals can also develop neurological and optical problems.  An estimated 1 in 1000 patients develop a form of encephalitis at the time of measles recovery. The acute encephalitis is severe inflammation of the brain that causes vomiting, convulsions, coma, brain damage, and even death \cite{aicardi1977acute}. Since measles is associated with a vitamin A deficiency, the virus can additionally place infected persons at higher risks for eye diseases, including xerophthalmia, corneal ulceration, keratomalacia, and subsequent blindness \cite{semba2004measles}. It is estimated that 15,000 of the 60,000 blindness cases reported each year among children in low income countries can be attributed to measles \cite{semba2004measles}.

\subsection{MMR Vaccine}

	 Finding an effective measles vaccine has been a main focus of many physicians and scientists for the past century. It wasn't until 1954 that biomedical scientist John Enders and physician T.C. Peebles were able to isolate a live but attenuated measles virus in tissue cultures in Boston \cite{1991}. This vaccine entered the United States market in 1963, but had the tendency to cause fever and rash in vaccinated individuals and was eventually replaced by the MMR vaccine in the 1970's \cite{Moss2012153}. 
	 
	 The MMR vaccine is a mixture of three vaccines that immunize against measles, mumps, and rubella (German measles) \cite{Linwood}. This vaccine has been clinically proven to be safe and costs less than 1 U.S. dollar per child \cite{dardis2012review}. The introduction of this combined vaccine has significantly reduced the occurrence of measles outbreaks in developed countries.  For the first 6 months of life, an infant usually possesses natural immunity from the disease due to maternal antibodies still present in the infant's system \cite{CBO9781139053518A160}. This passive immunity will wear off and it is therefore recommended to administer the MMR vaccine in two doses to optimize efficacy: the first when the recipient is between 12 and 15 months old and the second when the recipient is between 4 and 6 years old \cite{Moss2012153}. It has been found that the vaccine has approximately 90-95\% efficacy \cite{200215911620020301}.

\section{Epidemiological Overview}
\subsection{Measles Throughout History}
There have been multiple measles epidemics throughout history. American historian William McNeil claims that measles and other related diseases most likely originated in China sometime between A.D. 37 and A.D. 653 \cite{CBO9781139053518A160}. Since then, measles epidemics have continued to plague mankind. In the Middle Ages, many people confused the disease with smallpox \cite{CBO9781139053518A160}.  When the Europeans settled in North America during the fifteenth and sixteenth centuries, they unknowingly brought measles to the indigenous populations. Because these native populations had not yet developed any antibodies to help fight off the virus, many epidemics broke out and hundreds of thousands of Native Americans reportedly died over the course of several centuries \cite{Ennis}. Several South American Indian tribes in the Amazon were also lost, with the most notable epidemic causing 30,000 deaths in 1749 \cite{Weisberg2007471}. It wasn't until 1758 that physicians began to classify the disease as ``infectious" \cite{CBO9781139053518A160}. Despite many attempts to prevent and cure the disease, these epidemics continued to occur all around the world; during the American Civil War, 4,000 soldiers perished after becoming affected \cite{Ennis}. Outbreaks continued to occur until a weakened measles vaccine was introduced in the 1960's. The most current measles vaccine has substantially decreased the number of measles outbreaks occurring worldwide.

\subsection{Current Measles Outbreaks}
Despite the MMR vaccine available today, many countries are still dealing with the disease. It is estimated that only 50\% of all measles cases are actually reported to the World Health Organization.  This statistic indicates that every year the virus infects an estimated 50 million individuals, and as a result, causes 1.5 million deaths annually \cite{CBO9781139053518A160}. These numbers reflect the ease with which measles can re-infect a community even when only small pockets of the population are susceptible. In Quebec, Canada, where the average population immunity is estimated to be 95\%, an initial outbreak consisting of 94 cases transmitted through largely unrelated networks of unvaccinated individuals \cite{Moss2012153}. The same problem is exacerbated in developing countries where vaccination coverage is more sparse throughout the regions, namely in southern and eastern Africa. Out of the 46 African countries affected, recent measles outbreaks have been the most prevalent in South Africa, Zimbabwe, Zambia, and Malawi \cite{Moss2012153}.

Although measles was considered to be eliminated in the United States as of 2000, many European countries are still battling this disease.  The majority of outbreaks occur in Bulgaria, France, Italy, Germany, Ireland, the United Kingdom, and Spain \cite{cottrell2011measles}. In addition, these outbreaks are beginning to affect the U.S. again.  In 2011, the U.S. saw the highest number of annual measles cases in 2011 since 1995.  The Centers for Disease Control and Prevention primarily attribute these new cases to Americans traveling to Europe and bringing back the disease \cite{sepkowitz}.

\subsection{The Measles Epidemic in France}
After the MMR vaccine was instituted in France during the 1980's, the disease was practically nonexistent in the country. Unfortunately, the virus reappeared in 2008 \cite{freymuth2011measles}. The measles epidemic in France during 2008-2011 was the largest modern measles outbreak in Europe and is projected to increase even more during the next cycle \cite{cottrell2011measles}. Thus far, over 22,000 cases have been reported in the country, with 5,000 patients hospitalized from associated complications \cite{8586668720130301}. According to the World Health Organization, the measles strain originating from France has since traveled to Denmark, Gemany, Italy, Romania, Russia, and Belgium \cite{Green01072011}.
 \begin{figure}[h]
\begin{centering}
\includegraphics[scale=0.35]{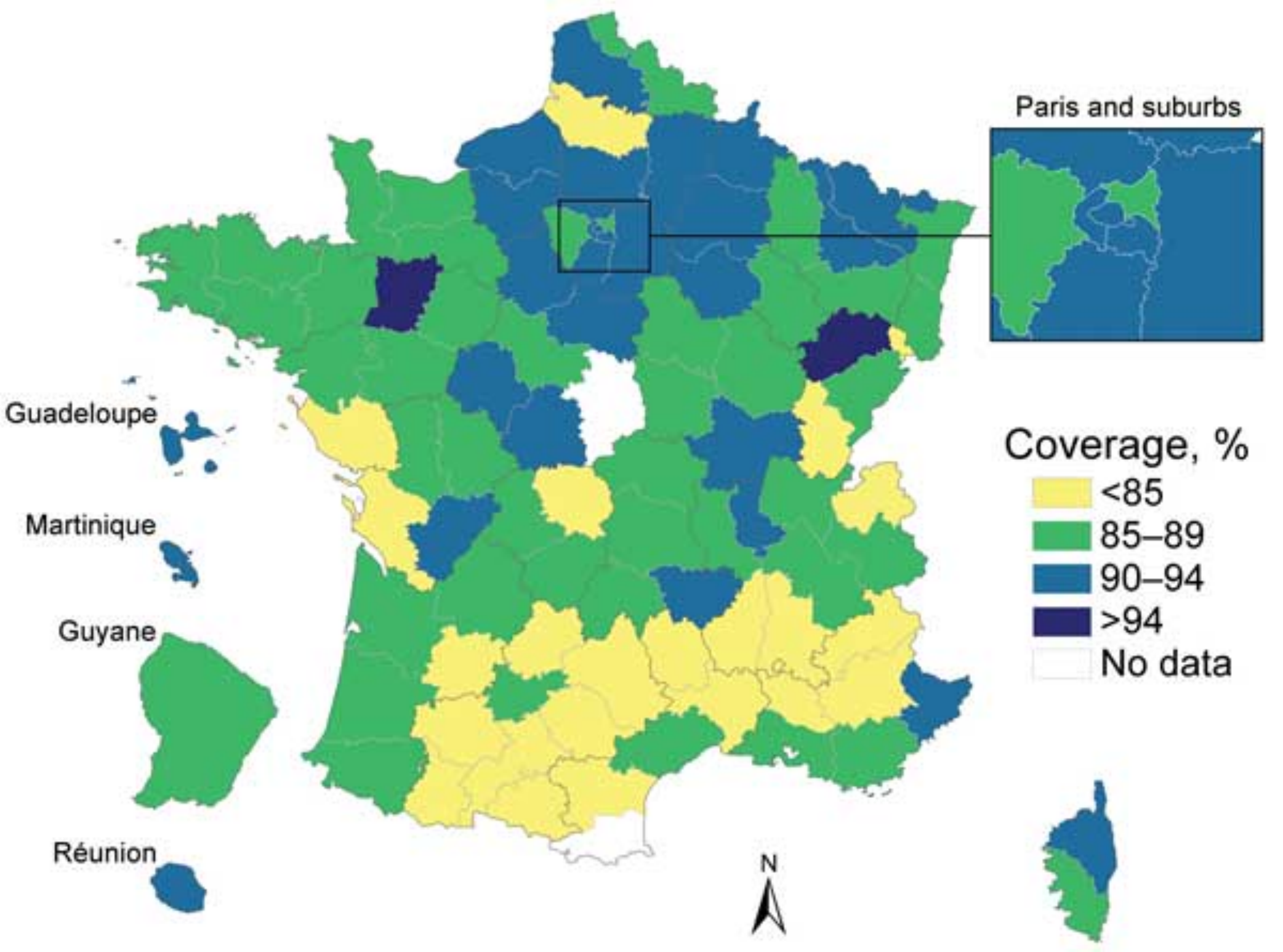}
\label{Fig3}
\caption{A map of France representing the MMR1 coverage rates in each district from 2003-3008. The data was collected from health certificates for children at 24 months of age. Figure reproduced from \cite{8586668720130301}.}
\end{centering}
\end{figure}
	
	Prior to the outbreak, the average immunization coverage in France fell below the recommended 95\% \cite{freymuth2011measles}. The primary cause for this average decrease was the lack of homogenous vaccination rates throughout the country; some regions had over 95\% of citizens vaccinated, while others had under 85\% of citizens vaccinated \cite{8586668720130301}. It has been noted that the districts containing the lowest vaccine coverage rates were located in southern France. As shown in Fig. 2, there are large disparities of vaccine coverage in each French district, suggesting that this heterogeneity may significantly impact the disease's transmission. 
\subsection{The Anti-Vaccine Trend}

Recently, a large portion of the global public has been losing confidence in the vaccine industry; this has had a significant impact on disease elimination efforts. For example, there have been multiple poliovirus outbreaks in Northern Nigeria due to the boycotting of the polio vaccine campaign \cite{moss2006global}. Many anti-vaccine proponents argue that the particular disease being vaccinated against is relatively mild and does not need vaccination coverage while others believe that the vaccine poses more of a danger than does the disease itself. Additionally, there are certain conservative religous sects that oppose vaccination and other modern health care methodologies.  Whatever the reason, new epidemics are occurring throughout the world as a result of this social trend.

	It has been speculated that the main cause for the wide disparity of vaccine coverage in France is not due to the vaccine's cost, but because many parents in particular districts are refusing to vaccinate their children.  When a fraudulent paper was published by Wakefield, et al. in 1998 that hypothesized an association between the MMR vaccine and autism, a widespread public fear of the vaccine was born.  Such unfounded fears were exacerbated in France when concerns surfaced around a hepatitis b vaccine. When these fears are combined with France's less aggressive vaccination policies, the anti-vaccination movement has fertile ground in which to grow \cite{moss2006global}. 
	
	We are motivated to construct a mathematical model for the spread of measles in France because we are interested in potential underlying forces inherent in the epidemic. We are specifically interested in the geographic heterogeneity of France because the epidemics correlate with low vaccination coverage in the southern regions. We aim to determine the impact of heterogenous vaccination coverage on epidemics to gain insight on disease dynamics within developed and developing countries. 

\section{Model}
The outbreaks in France and the heterogeneous nature of vaccine coverage motivated us to study a two patch SEIR model. The country is divided by regions into separate populations based on the vaccination rate given by \cite{8586668720130301}. We justify the use of an SEIR model because measles has a latent period where patients are infected but not infective. Additionally, we assume that the population is born either susceptible or immune based on a combination the vaccination rate of the regions and the expected efficacy of the vaccine. The recovered class also includes those who are vaccinated and imbues perfect immunity against measles. Finally, we assume that birth rate of the population is equal to the death rate, and infection with measles does not increase death rate.\\

\subsection{General Model}

Generally, we construct our multi-regional model with four components: susceptible individuals that are unvaccinated, exposed individuals in the latent period of the disease, infectious individuals, and resistant individuals that are either recovered or vaccinated,

\begin{eqnarray*}
	\dot{S}_i&=&-S_iF_i(I_i,I_j)-\mu S_i+(1-\delta_i)\mu N_i,\\
	\dot{E}_i&=&S_iF_i(I_i,I_j)-(\mu+\phi)E_i,\\
	\dot{I}_i&=&\phi E_i-(\mu+\gamma)I_i,\\
	\dot{R}_i&=&\gamma I_i -\mu R_i + \delta_i\mu N_i,\\
\end{eqnarray*}

where $i=1,2,$ $j=1,2,$ and $i\neq j.$ The function $F_i$ models the mixing of the two patches and includes the contact rate and probability of infection given a successful contact. The function $F_i$ can take many forms depending on biological assumptions. Simple heterogenous mixing as discussed in \cite{brauer2008epidemic} is a reasonable assumption to begin with. Proportional mixing assumes that every individual across the entire population associates with one another solely based upon that individual's activity level. Preferential mixing assumes that a subsection of each patch only mixes with members of the same patch, while the remaining patch population mixes proportionally with the populations of all patches \cite{nold1980heterogeneity, brauer2008epidemic}.  While the proportional mixing assumptions could accurately represent vaccine coverage disparities, it fails to represent the natural patch isolation that occurs in such a large geographical region as France.  Therefore, we chose to use a preferential mixing model.  Such a model allows the representation of heterogeneous vaccination and activity rates as well as the natural patch preference that occurs when considering a country regionally.  Preferential mixing in the general form expressed by \cite{brauer2008epidemic} then gives the function:
\begin{center}
$\displaystyle F_i(I_i,I_j)=a_i\left(p_{ii}\frac{I_i}{N_i}+\sum_{j\neq i} p_{ij}\frac{I_j}{N_j}\right),$
\end{center}
where $a_i$ is representative of number of contacts that individuals make in a given time and the likelihood there would be a successful transmission given the contact was with an infective.  The proportion of contacts that susceptibles make with people in their own patch is expressed with the term $p_{ii}\frac{I_i}{N_i},$ and the proportion of contacts with members of other patches is expressed in the term $\sum p_{ij}\frac{I_j}{N_j}.$

Thus the full form of our system given vaccination and preferential mixing is given by:

\begin{eqnarray*}
	\dot{S}_i&=&-a_iS_i\left(p_{ii}\frac{I_i}{N_i}+\sum p_{ij}\frac{I_j}{N_j}\right)-\mu S_i+(1-\delta_i)\mu N_i,\\
	\dot{E}_i&=&a_iS_i\left(p_{ii}\frac{I_i}{N_i}+\sum p_{ij}\frac{I_j}{N_j}\right)-(\mu+\phi)E_i,\\
	\dot{I}_i&=&\phi E_i-(\mu+\gamma)I_i,\\
	\dot{R}_i&=&\gamma I_i -\mu R_i + \delta_i\mu N_i.\\
\end{eqnarray*}
where $i=1,2, j\neq i.$

\subsection{Positive Invariance and  Boundedness}
First we will show upper boundedness of our model.  As no real biological system would blow up to infinity, it is important to show that any model is bounded above.  To do so consider that $N_i=S_i+E_i+I_i+R_i$ and $\dot{N}_i=\dot{S}_i+\dot{E}_i+\dot{I}_i+\dot{R}_i$ which becomes $\dot{N}_i = \mu N_i - \mu (S_i + E_i + I_i + R_i)  = \mu N - \mu N = 0$.  Therefore, as $\dot{N}_i=0$ for all time the population is constant for all time.  Thus $N(t)=N_i(0)$ for all time, and is bounded by the initial condition for all time.  This immediately implies the boundedness of all class variables to be in the interval $[0,~N_i(0)]$.

Positive invariance guarantees that the model biologically is well posed and will not create illogical solutions and negative populations.  We will now show positive invariance of the system given initial conditions $(S_i(0),E_i(0),I_i(0),R_i(0))>0$ and $N_i=S_i+E_i+I_i+R_ i=N_i(0)$.  Towards contradiction we assume $\exists t_s>t_0$ such that the first zero point is $S_i(t_s)=0$.  Now consider $\dot{S}_i(t_s)|_{S_i(t_s)=0}=(1-\delta_i)N_i>0$ because $\delta_i < 1$ and $N_i > 0$ given $N_i(0) > 0$.  This implies that $\exists t_c$ such that $0 \leq t_c \leq t_s$ and $S_i(t_c)<0$.  However, given that $S_i(0)>0$ and $S_i(t_s)=0$ is the first zero point, we have a contradiction.  Therefore by proof through contradiction, $S_i>0$ $\forall t\in[0,\infty)$. 

For all other state variables the proof follows similarly.

%$\exists t_e>t_0$ such that $E_i(t_e)=0$ and $\dot{E}_i(t_e)<0$.  However, it is the case that $\dot{E}_i(t_e)|_{E_i(t_e)=0}=S_iF_i(I_i,I_j)\geq0$ since $S_i\geq0$.  Therefore by proof through contradiction, $E_i\geq0$ $\forall t\in[0,\infty)$. 
%
%$\exists t_i>t_0$ such that $I_i(t_i)=0$ and $\dot{I}_i(t_i)<0$.  However, it is the case that $\dot{I}_i(t_i)|_{I_i(t_i)=0}=\phi E_i\geq0$ since $E_i\geq 0$.  Therefore by proof through contradiction, $I_i\geq0$ $\forall t\in[0,\infty)$. 
%
%$\exists t_r>t_0$ such that $R_i(t_r)=0$ and $\dot{R}_i(t_r)<0$.  However, it is the case that $\dot{R}_i(t_r)|_{R_i(t_r)=0}=\gamma I_i+\delta_iN_i\geq0$ since $I_i\geq0$.  Therefore by proof through contradiction, $R_i\geq0$ $\forall t\in[0,\infty)$.\\

Positive invariance has been proven and have now shown all state variables are bounded both above and below. Therefore given $(S_i(0),E_i(0),I_i(0),R_i(0),N_i(0))>0$ all classes and total population $0<(S_i,E_i,I_i,R_i,N_i )\leq N_i(0)  ~~~\forall t\in[0,\infty)$ and the model is well posed. 

\subsection{Rescaling}
Starting from a multi-patch mixing model with vaccination, we examine how a special case of preferential mixing can be expressed for a two-patch system.

\begin{eqnarray*}
	\dot{S}_i&=&-a_iS_i\left(p_{ii}\frac{I_i}{N_i}+\sum p_{ij}\frac{I_j}{N_j}\right)-\mu S_i+(1-\delta_i)\mu N_i,\\
	\dot{E}_i&=&a_iS_i\left(p_{ii}\frac{I_i}{N_i}+\sum p_{ij}\frac{I_j}{N_j}\right)-(\mu+\phi)E_i,\\
	\dot{I}_i&=&\phi E_i-(\mu+\gamma)I_i,\\
	\dot{R}_i&=&\gamma I_i -\mu R_i + \delta_i\mu N_i,\\
\end{eqnarray*}

where
	
\begin{eqnarray*}
	p_{ii}=p_{11}&=&\pi_i+(1-\pi_i)p_1,\\
	p_{ij}=p_{12}&=&(1-\pi_i)p_2,\\
	p_1=p_2&=&\frac{(1-\pi_i)a_iN_i}{(1-\pi_1)a_1N_2+(1-\pi_2)a_2N_2}.\\
\end{eqnarray*}
\\
For simplifying purposes, assume that $a_i=a_j,$ $N_i=N_j$ and $\pi_i=\pi_j$;
\begin{eqnarray*}
	\frac{(1-\pi)aN}{(1-\pi)aN+(1-\pi)aN}&=&\frac{(1-\pi)}{(1-\pi)+(1-\pi)},\\
	p_1=p_2&=&\frac{(1-\pi)}{(1-\pi)+(1-\pi)}=\frac{1}{2},\\
	p_{11}=\pi+(1-\pi)p_1&=&\pi+\frac{1-\pi}{2},\\
	p_{12}=(1-\pi)p_2&=&\frac{1-\pi}{2}.\\
\end{eqnarray*}	
\\
Therefore the system becomes:
\begin{eqnarray*}
\dot{S}_i&=&-aS_i\left(\left(\pi+\frac{1-\pi}{2}\right)\frac{I_i}{N_i}+\left(\frac{1-\pi}{2}\right)\frac{I_j}{N_j}\right)-\mu S_i+(1-\delta_i)\mu N_i,\\
	\dot{E}_i&=&aS_i\left(\left(\pi+\frac{1-\pi}{2}\right)\frac{I_i}{N_i}+\left(\frac{1-\pi}{2}\right)\frac{I_j}{N_j}\right)-(\mu+\phi)E_i,\\
	\dot{I}_i&=&\phi E_i-(\mu+\gamma)I_i,\\
	\dot{R}_i&=&\gamma I_i -\mu R_i + \delta_i\mu N_i.\\
\end{eqnarray*}
\\
Now consider the term;
\begin{eqnarray*}
	aS_i\left(\left(\pi+\frac{1-\pi}{2}\right)\frac{I_i}{N_i}+\left(\frac{1-\pi}{2}\right)\frac{I_j}{N_j}\right),\\
	p_{11}=\pi+\frac{1-\pi}{2}=\frac{2\pi}{2}+\frac{1-\pi}{2}=\frac{1+\pi}{2},\\
	1-p_{12}=1-\frac{1-\pi}{2}=\frac{2}{2}-\frac{1-\pi}{2}=\frac{1+\pi}{2}.\\
\end{eqnarray*}
\\
This allows the following rescaling;
\begin{eqnarray*}
	\rho=p_{12}=\frac{1-\pi}{2},\\
	(1-\rho)=p_{11}=\frac{1+\pi}{2},\\
	aS_i\left((1-\rho)\frac{I_i}{N_i}+\rho\frac{I_j}{N_j}\right),\\
	aS_i(1-\rho)\left(\frac{I_i}{N_i}+\left(\frac{\rho}{1-\rho}\right)\frac{I_j}{N_j}\right),\\
	\beta S_i\left(\frac{I_i}{N_i}+\alpha\frac{I_j}{N_j}\right),\\
\end{eqnarray*}
where
\begin{eqnarray*}
	\beta=a(1-\rho),\\
	\alpha=\frac{\rho}{1-\rho}.\\
\end{eqnarray*}
\\
The model is constructed to have a constant population invariant with respect to time, i.e. where $N_i=S_i+E_i+I_i+R_i$ and the time derivative of $N_i$ is $\dot{N}_i=\dot{S}_i+\dot{E}_i+\dot{I}_i+\dot{R}_i=0.$ In addition, the $R_i$ class does not impact the dynamics the system, thus the two-patch model can be reduced to a $6$ dimensional system by $R_i=N_i-S_i-E_i-I_i.$  In addition, we divide each variable by its respective patch population size, $(s_1,e_1,i_1,s_2,e_2,i_2)=\left(\frac{S_1}{N_1},\frac{E_1}{N_1},\frac{I_1}{N_1},\frac{S_2}{N_2},\frac{E_2}{N_2},\frac{I_2}{N_2}\right).$ Thus the rescaled two-patch system can be written as:

\begin{eqnarray*}
	\dot{s}_1&=&-\beta s_1(i_1+\alpha i_2)-\mu s_1+(1-\delta_1)\mu,\\
	\dot{e}_1&=&\beta s_1(i_1+\alpha i_2)-(\mu + \phi)e_1,\\
	\dot{i}_1&=&\phi e_1-(\mu + \gamma)i_1,\\
	\dot{s}_2&=&-\beta s_2(i_2+\alpha i_1)-\mu s_2+(1-\delta_2)\mu,\\
	\dot{e}_2&=&\beta s_2(i_2+\alpha i_1)-(\mu + \phi)e_2,\\
	\dot{i}_2&=&\phi e_2-(\mu + \gamma)i_2,\\
\end{eqnarray*}
which shall be investigated in the sections to follow.

\section{Analysis}
\subsection{Basic Reproductive Number and the Disease Free Equilibrium} 
In epidemiological models, the basic reproductive number $(\mathcal{R}_0)$ is an important element in analysis.  In brief, $\mathcal{R}_0$ represents the number of new cases that stem from an initial infective within an entirely susceptible population.  It is often the case that if $\mathcal{R}_0<1$ the disease will die out and if $\mathcal{R}_0>1$ the disease will persist.  However, this is not always the case in more complicated models.  For models that account for vaccination, it is also necessary to consider the reproduction number of the disease given the controls $(\mathcal{R_C}).$  \\

In our two-patch system we find both $\mathcal{R}_0$ and $\mathcal{R_C}$ using the next generation operator method as outlined in \cite{van2002reproduction}.  This allows for the restructuring of the model into a vector equation and finding the next generation matrix, and evaluating at the disease free equilibrium.  The spectral radius, (largest eigenvalue), of this matrix is equivalent to $\mathcal{R_C}$ for the system.\\

The disease free equilibria (DFE) occurs when a fixed point in any epidemiological system is such that all disease-carrying classes are zero.  For our rescaled two-patch model, the absence of disease occurs when $e^*_1=e^*_2=i^*_1=i^*_2=0,$ $s_1=s_1^*,$ and $s_2=s_2^*.$  (Here, $x_i^*$ indicates a fixed value of $x_i$). Given the necessary absence of disease and the assumption of the existence of some $s_1=s_1^*$ and $s_2=s_2^*,$ we look for a point which satisfies,\\
	\begin{eqnarray*}
	\dot{s}_1&=&-\beta s_1(i_1+\alpha i_2)-\mu s_1+(1-\delta_1)\mu=0,\\
	\dot{e}_1&=&\beta s_1(i_1+\alpha i_2)-(\mu + \phi)e_1=0,\\
	\dot{i}_1&=&\phi e_1-(\mu + \gamma)i_1=0,\\
	\dot{s}_2&=&-\beta s_2(i_2+\alpha i_1)-\mu s_2+(1-\delta_2)\mu=0,\\
	\dot{e}_2&=&\beta s_2(i_2+\alpha i_1)-(\mu + \phi)e_2=0,\\
	\dot{i}_2&=&\phi e_2-(\mu + \gamma)i_2=0.\\
\end{eqnarray*}

We then find the DFE by considering

\begin{eqnarray*}
-\mu s_1^*+(1-\delta_1)\mu=0,\\
-\mu s_2^*+(1-\delta_2)\mu=0.\\
\end{eqnarray*}

From this, we determine the disease free equilibrium to exist at

\begin{eqnarray*}
	s_1^*=1-\delta_1,\\
	s_2^*=1-\delta_2,\\
	e_1^*=e_2^*=0,\\
	i_1^*=i_2^*=0.\\
\end{eqnarray*}

To find the next generation matrix, we express the system as a set of vector equations where $X$ is the vector of infected classes and $Y$ is the vector of uninfected classes. Thus, we have
\begin{eqnarray*}
\frac{dX}{dt} &=& \mathcal{F}(X,Y) - \mathcal{V}(X,Y), \\
\frac{dY}{dt} &=& \mathcal{W}(X,Y).\\
\end{eqnarray*}
Here, $\mathcal{F}(X,Y)$ represents flows from $Y$ into $X,$ and $\mathcal{V}(X,Y)$ represents all other flows. We then set 

\begin{eqnarray*}
F = \left(\frac{\partial \mathcal{F}}{\partial X} \right) _{(DFE)}, V = \left(\frac{\partial \mathcal{V}}{\partial X}\right) _{(DFE)}.\\
\end{eqnarray*}
For our system, one gets

\begin{eqnarray*}
\mathcal{F}(X,Y) = \left[ \begin{array}{c}
				\beta s_1(i_1 +\alpha i_2) \\
				0\\
				\beta s_2(i_2 +\alpha i_1) \\
				0\\
				\end{array} \right] &,&			
\mathcal{V}(X,Y) = \left[ \begin{array}{c}
				e_1(\mu + \phi) \\
				-\phi e_1 + i_1(\mu + \gamma) \\
				e_2 (\mu + \phi) \\
				-\phi e_2 + i_2(\mu + \gamma) \\
				\end{array} \right], \\			
F_{(DFE)} = \left[ \begin{array}{c  c  c  c}
			0 & \beta(1-\delta_1) & 0 & \beta(1-\delta_1)\alpha \\
			0 & 0 & 0 & 0 \\
			0 & \beta(1-\delta_2)\alpha & 0 & \beta(1-\delta_2) \\
			0 & 0 & 0 & 0\\
			\end{array} \right] &,&
V_{(DFE)} = \left[ \begin{array}{c c c c}
			\mu + \phi & 0 & 0 & 0\\
			-\phi & \mu+\gamma & 0 & 0\\
			0 & 0 & \mu+\phi & 0\\
			0 & 0 & -\phi & \mu+\gamma \\
			\end{array}\right]. \\
\end{eqnarray*}
We can compute the next generation matrix $FV^{-1},$ where the spectral radius of $FV^{-1}$ is the control reproduction number $\mathcal{R_C},$ and the case without vaccination is the basic reproduction number $\mathcal{R}_0.$ As a result,

$$\mathcal{R_C} =\displaystyle \frac{\beta \phi \left((1-\delta_1)+(1-\delta_2) + \sqrt{4 \alpha^2 (1-\delta_1)(1-\delta_2) + (\delta_1 - \delta_2)^2}\right)}{2(\mu+\phi)(\mu+\gamma)},$$
$$\mathcal{R}_0 = \displaystyle \frac{\beta \phi (1+\alpha)}{(\mu+\phi)(\mu+\gamma)}.$$

From literature, (\cite{keeling2002understanding}, \cite{bauch2003transients}, \cite{bolker1993chaos}), we find $\mathcal{R}_0$ to be commonly considered between 12 and 16. We use this information in our estimation of other parameters as seen in Appendix A.

From theorems of the next generation operator and definitions of $\mathcal{R}_0$ and $\mathcal{R_C}$ we also know immediately that the disease free equilibrium is stable when $\mathcal{R_C}<1.$  Therefore no further analysis is needed.  It is the case that if
$$\mathcal{R_C} =\displaystyle \frac{\beta \phi \left((1-\delta_1)+(1-\delta_2) + \sqrt{4 \alpha^2 (1-\delta_1)(1-\delta_2) + (\delta_1 - \delta_2)^2}\right)}{2(\mu+\phi)(\mu+\gamma)}<1$$
the DFE is stable.

\subsection{Endemic Equilibrium} 
Solving for the endemic equilibrium explicitly, $(s_1^*,i_1^*\neq0,e_1^*\neq0,s_2^*,e_2^*\neq0,i_2^*\neq0),$ proved analytically formidable. Instead, we adopted the method of solving the equations systematically to produce expressions for the equilibrium points in terms of the exposed classes. then

\begin{eqnarray*}
s_1^* &=& \frac{\mu(1-\delta_1)(\mu+\gamma)}{ \beta \phi(e_1^* +\alpha e_2^*) + \mu(\mu+\gamma)},\\
e_1^* &=& \frac{e_2^*(\mu(\beta \phi (1-\delta_2) - (\mu+\phi)(\mu+\gamma)) - e_2^* \beta \phi (\mu+\phi))}{\alpha \beta \phi (e_2^* (\mu + \phi) - \mu (1-\delta_2))},\\
i_1^* &=& \frac{e_1^* \phi}{\mu+\gamma},\\
s_2^* &=& \frac{\mu(1-\delta_2)(\mu+\gamma)}{\beta \phi(e_2^* +\alpha e_1^*)  + \mu(\mu+\gamma)},\\
e_2^* &=& \frac{e_1^*(\mu(\beta \phi (1-\delta_1) - (\mu+\phi)(\mu+\gamma)) - e_1^* \beta \phi (\mu+\phi))}{\alpha \beta \phi (e_1^* (\mu + \phi) - \mu (1-\delta_1))},\\
i_2^* &=& \frac{e_2^* \phi}{\mu+\gamma}.\\
\end{eqnarray*}

Upon substituting $D_1 = 1-\delta_1, D_2 = 1=\delta_2, A = \mu+\gamma, $ and $ B=\mu+\phi,$ our endemic equilibrium equations become

\begin{eqnarray*}
s_1^* &=& \frac{\mu D_1A}{\beta \phi(e_1^* +\alpha e_2^*)  + \mu A},\\
e_1^* &=& \frac{e_2^*(\mu(\beta \phi D_2 -AB) - e_2^* \beta \phi B)}{\alpha \beta \phi (e_2^*B - \mu D_2)},\\
i_1^* &=& \frac{e_1^* \phi}{A},\\
s_2^* &=& \frac{\mu D_2 A}{\beta \phi(e_2^* +\alpha e_1^*)  + \mu A},\\
e_2^* &=& \frac{e_1^*(\mu(\beta \phi D_1 - AB) - e_1^* \beta \phi B)}{\alpha \beta \phi (e_1^* B - \mu D_1)},\\
i_2^* &=& \frac{e_2^* \phi}{A}.\\
\end{eqnarray*}
Notice that all potential equilibrium values, namely $s_1^*,i_1^*,s_2^*,i_2^*,$ are positive if and only if $e_1^*>0$ and $e_2^*>0.$ Thus we analyze the equations for $e_1^*$ and $e_2^*$ for conditions under which these variables are positive. The results are presented in the following statements.\\

\begin{thm}
If the endemic equilibrium exists in $\mathbb{R}^+$ then $e_i^*$ is bounded.
\end{thm}

\begin{proof}
W.L.O.G., consider the equation for $e_i^*$ from the endemic equilibrium: 
$$e_i^*=\frac{e_j^*(\mu(D_j\beta\phi-AB)-B\beta\phi e_j^*)}{\alpha\beta\phi(Be_j^*-D_j\mu)}$$ where $i\neq j.$ Assume that the denominator is positive, 
$$\alpha\beta\phi(Be_j^*-D_j\mu)>0$$ 
Then it must hold that,
$$Be_j^*>D_j\mu\Rightarrow e_j^*>\frac{D_j\mu}{B}$$ 
The numerator must then also be positive for the endemic equilibrium to exist in $\mathbb{R}^+,$ 
$$e_j^*(\mu(D_j\beta\phi-AB)-B\beta\phi e_j^*)>0$$ 
which is only possible if two conditions hold: 
$$D_j\beta\phi>AB$$ 
and 
$$\mu(D_j\beta\phi-AB)>B\beta\phi e_j^*.$$
If we assume $D_j\beta\phi>AB,$ then: 
$$\mu(D_j\beta\phi-AB)>B\beta\phi e_j^*$$
$$\Rightarrow\frac{D_j\mu}{B}-\frac{\mu A}{\beta\phi}>e_j^*>\frac{D_j\mu}{B},$$
which is a contradiction since $\frac{\mu A}{\beta\phi}>0.$ Then the only possibility is that $e_i^*<\frac{D_j\mu}{B}.$ Thus if the the endemic equilibrium exists in $\mathbb{R}^+,$ $e_i^*$ must be bounded above by $\frac{D_j\mu}{B}.$
\end{proof}

\begin{thm}
If $e_j^*<\frac{D_j\mu}{B},$ then the endemic equilibrium always exists in $\mathbb{R}^+.$
\end{thm}

\begin{proof}
W.L.O.G., consider the equation for $e_i^*$ from the endemic equilibrium: 
$$e_i^*=\frac{e_j^*(\mu(D_j\beta\phi-AB)-B\beta\phi e_j^*)}{\alpha\beta\phi(Be_j^*-D_j\mu)}$$ where $i\neq j.$ When $e_j^*<\frac{D_j\mu}{B},$ the denominator is negative. Thus for the endemic equilibrium to exist in $\mathbb{R}^+$ the numerator must also be negative. We have two options:

Case 1. $D_j\beta\phi<AB.$ Then the numerator is always negative.

Case 2. $D_j\beta\phi>AB.$ The numerator is only negative if $\mu(D_j\beta\phi-AB)<B\beta\phi e_j^*.$ Indeed,
$$\mu(D_j\beta\phi-AB)<B\beta\phi e_j^*$$
$$\Rightarrow\frac{D_j\mu}{B}-\frac{\mu A}{\beta\phi}<e_j^*<\frac{D_j\mu}{B}$$
This always holds if $e_j^*<\frac{D_j\mu}{B},$ thus under this condition, the endemic equilibrium always exists in $\mathbb{R}^+.$
\end{proof}

\begin{thm}
If $D_j\beta\phi>AB$ then $\mathcal{R_C}>1$
\end{thm}

\begin{proof}
Recall that 
$$\mathcal{R_C}=\frac{\beta\phi\left((1-\delta_1)+(1-\delta_2)+\sqrt{(\delta_1-\delta_2)^2+4\alpha^2(1-\delta_1)(1-\delta_2)}\right)}{2(\mu+\gamma)(\mu+\phi)},$$
which, under the above algebraic simplifications becomes 
$$\mathcal{R_C}=\frac{\beta\phi\left(D_1+D_2+\sqrt{(D_2-D_1)^2+4\alpha^2D_1D_2}\right)}{2AB}.$$
If we have $D_j\beta\phi>AB$ then 
$$\mathcal{R_C}=\frac{\beta\phi\left(D_1+D_2+\sqrt{(D_2-D_1)^2+4\alpha^2D_1D_2}\right)}{2AB}>\frac{2AB+\beta\phi\sqrt{(D_2-D_1)^2+4\alpha^2D_1D_2}}{2AB}>1.$$
Thus if $D_j\beta\phi>AB,$ $\mathcal{R_C}>1.$
\end{proof}

\begin{thm}
If $D_1=D_2\equiv D$ then $D<\frac{AB}{\beta\phi}\Rightarrow \mathcal{R_C}<1+\alpha$ and $D>\frac{AB}{\beta\phi}\Rightarrow \mathcal{R_C}>1.$
\end{thm}

\begin{proof}
When $D_1=D_2\equiv D$ we can show that 

\begin{eqnarray*}
\mathcal{R_C}&=&\frac{\beta\phi\left(D_1+D_2+\sqrt{(D_2-D_1)^2+4\alpha^2D_1D_2}\right)}{2AB}\\
&=&\frac{\beta\phi\left(2D+\sqrt{4\alpha^2D^2}\right)}{2AB}=\frac{\beta\phi\left(D+\alpha D\right)}{AB}=\frac{\beta\phi}{AB}D(1+\alpha).
\end{eqnarray*}
If $D<\frac{AB}{\beta\phi}$ then 
$$\mathcal{R_C}=\frac{\beta\phi}{AB}D(1+\alpha)<\frac{\beta\phi}{AB}\frac{AB}{\beta\phi}(1+\alpha)=1+\alpha.$$
If $D>\frac{AB}{\beta\phi}$ then 
$$\mathcal{R_C}=\frac{\beta\phi}{AB}D(1+\alpha)>\frac{\beta\phi}{AB}\frac{AB}{\beta\phi}(1+\alpha)=1+\alpha\geq1.$$
\end{proof}

\begin{cor}
If $D_1=D_2\equiv D$ then $D<\frac{AB}{\beta\phi(1+\alpha)}=\left(\frac{1}{\mathcal{R}_0}\right)\Rightarrow\mathcal{R_C}<1$
\end{cor}

\subsubsection{Stability of the Endemic Equilibria}
Numerical methods and algorithms derived from the above theorems allowed for both the existence and stability of the endemic equilibrium to be calculated.\\

\begin{figure}[!h]
\centering
\includegraphics[scale=0.2]{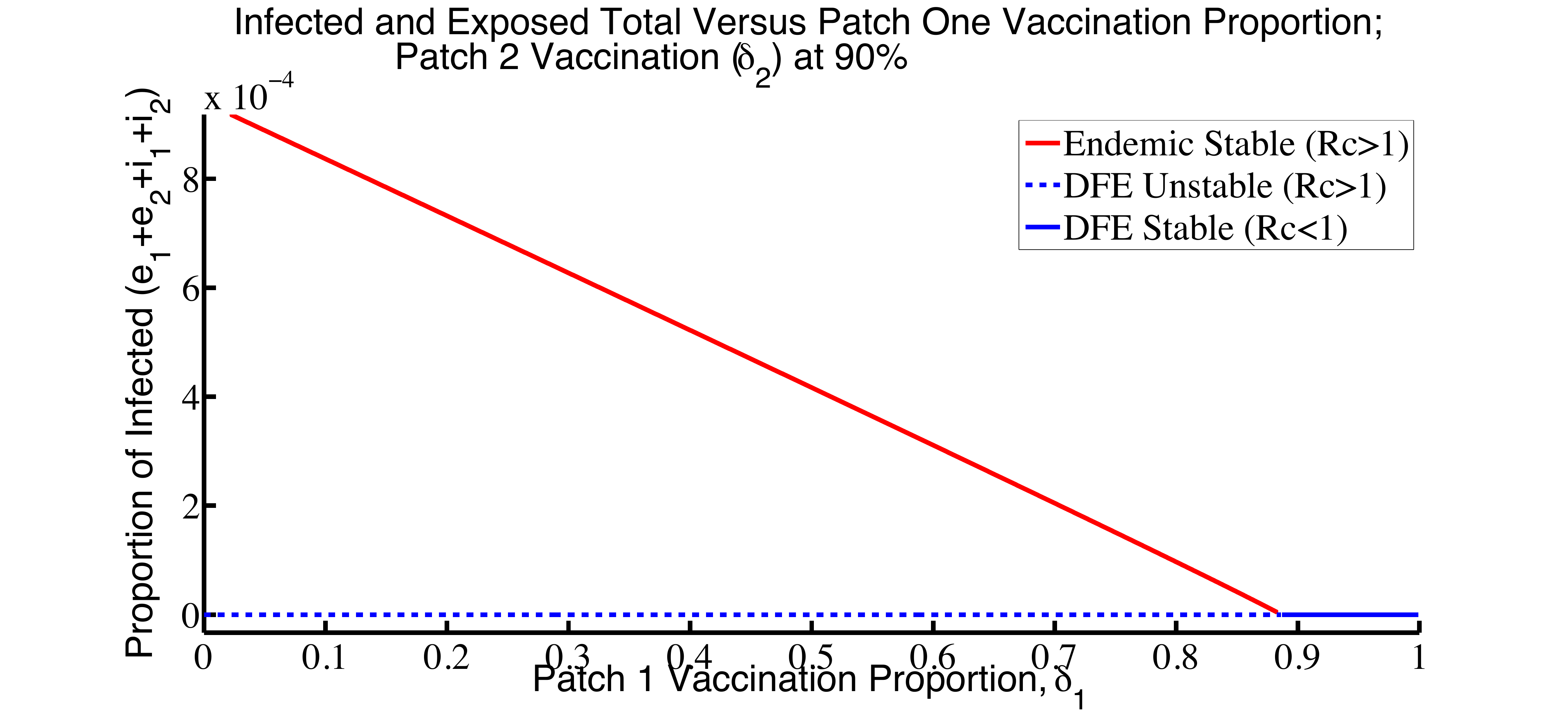}
\caption{Bifurcation diagram showing the existence and stability of the endemic equilibria as a function of $\delta_1$ with all other parameters fixed.}
\end{figure}

While there is not an explicit closed form solution, it is possible to determine existence and stability numerically given some parameter space.\\

\subsection{$\mathcal{R_C}$ and Vaccination Heterogeneity}
We found that $\mathcal{R_C}$ in the case of two patches with heterogeneous vaccination coverage is greater than or equal to $\mathcal{R_C}$ with homogeneous vaccination coverage.   Using the previously defined substitutions, $\mathcal{R_C}$ given $D_1 \neq D_2$ is greater than or equal to $\mathcal{R_C}$ for two patches with the same vaccination coverage, $D_*:=\frac{D_1+D_2}{2}.$ This means that for a given average vaccination proportion for an entire population, having heterogeneous vaccination coverage increases $\mathcal{R_C}.$ It is now necessary to demonstrate the following inequality:\\

$$\mathcal{R_C}\left(D_1,D_2 | D_1\neq D_2\right) \geq \mathcal{R_C}\left(\frac{D_1+D_2}{2},\frac{D_1+D_2}{2}\right).$$\\
\\
For $\mathcal{R_C}\left(D_1,D_2 | D_1\neq D_2\right),$ we have $$\mathcal{R_C} = \displaystyle \frac{\beta \phi (D_1+D_2+\sqrt{(D_1-D_2)^2+4\alpha^2D_1D_2})}{2AB},$$ and for the case $\displaystyle\mathcal{R_C}\left(\frac{D_1+D_2}{2},\frac{D_1+D_2}{2}\right)$ $\mathcal{R_C}$ reduces to $$\mathcal{R_C} = \displaystyle \frac{\beta \phi (D_1+D_2)(1+\alpha)}{2AB}.$$ Therefore it is sufficient to show that $\sqrt{(D_1-D_2)^2+4\alpha^2D_1D_2}) \geq \alpha(D_1+D_2).$ \\
\\
By definition, we have $0 \leq \alpha \leq 1$ and $(D_1-D_2)^2 \geq 0,$ so then $(1-\alpha^2)(D_1-D_2)^2 \geq 0.$ By expanding the term completely and moving all $\alpha^2$ terms to one side, we have:

$$D_1^2-2D_1D_2+D_2^2 \geq \alpha^2D_1^2 - 2\alpha^2D_1D_2 + \alpha^2D_2^2.$$
\\
Adding $4\alpha^2D_1D_2$ to both sides gives:
$$D_1^2-2D_1D_2+D_2^2+4\alpha^2D_1D_2 \geq \alpha^2D_1^2 + 2\alpha^2D_1D_2 + \alpha^2D_2^2,$$
\\
or
$$(D_1-D_2)^2+4\alpha^2D_1D_2 \geq \alpha^2(D_1+D_2)^2.$$
\\
The inequality is preserved if we take the square root of both sides. Then we have:
$$\sqrt{(D_1-D_2)^2+4\alpha^2D_1D_2} \geq \alpha(D_1+D_2).$$\\

\begin{figure}
\centering
\includegraphics[scale=0.2]{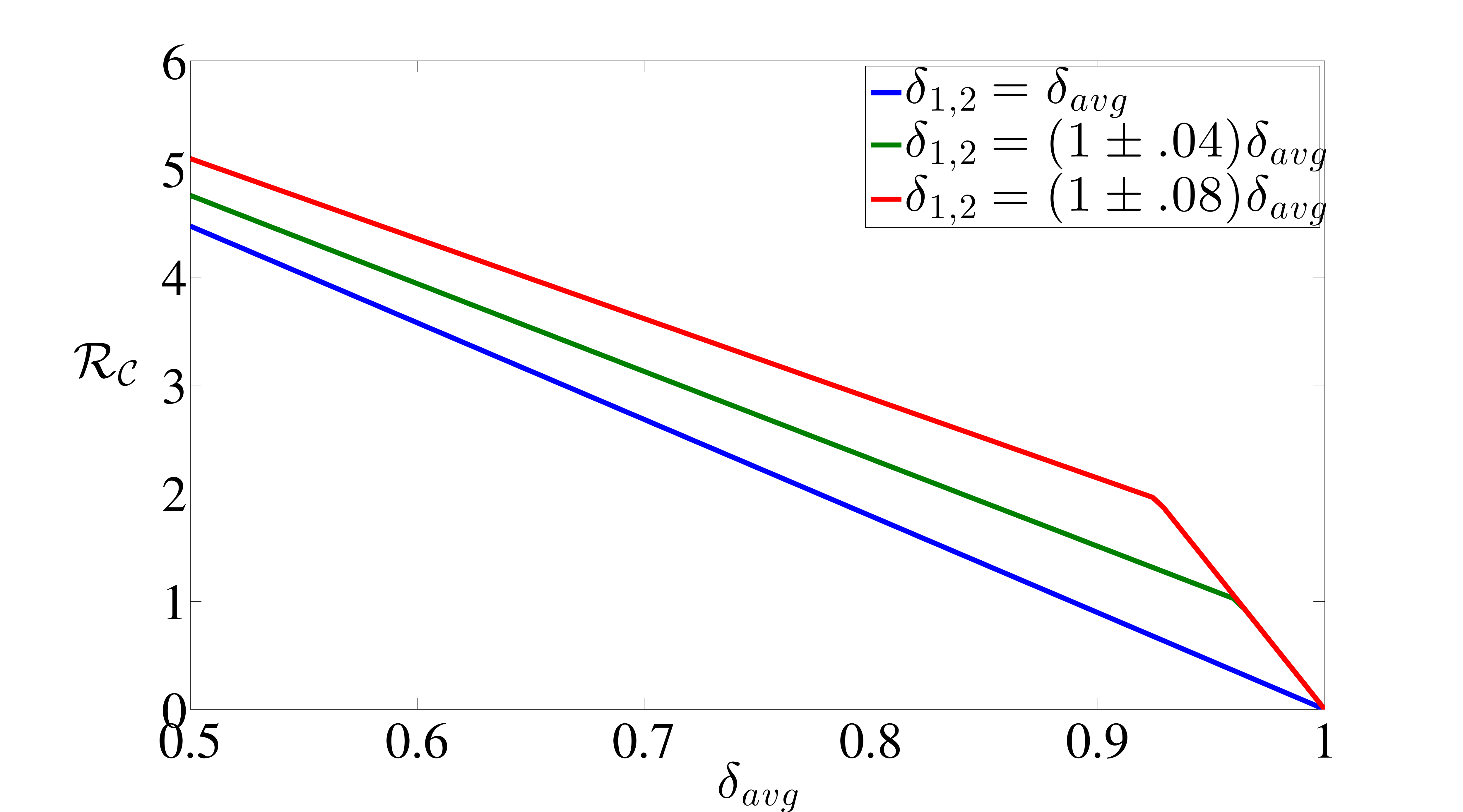}
\caption{Comparing two-patch $\mathcal{R_C}$ values in a homogenous case and two heterogeneous cases.  The deflection is such that for an average of 1, both patches must be 1 due to the constraints on $\delta$; all values converge along that line.}
\end{figure}

This completes the proof and shows that given an average vaccination coverage for a population, heterogeneous coverage causes $\mathcal{R_C}$ to be greater than or equal to homogeneous coverage.
%\subsection{Conditions for $\mathcal{R_C}<1$}
By common definition, $\mathcal{R_C}$ is a control reproductive number given some intervention.  As such, it is worthwhile to determine what level of intervention is required in order to drive $\mathcal{R_C}<1.$  It was previously shown in the endemic equilibria analysis that if $D_1=D_2\equiv D$ then $D<\frac{AB}{\beta\phi(1+\alpha)}=\left(\frac{1}{\mathcal{R_C}}\right)\Rightarrow\mathcal{R_C}<1.$  Recalling that $D=1-\delta$ this gives:
\begin{center}
If $D_1=D_2\equiv D$ then $\delta>1-\frac{AB}{\beta\phi(1+\alpha)}=1-\frac{1}{\mathcal{R}_0}\Rightarrow\mathcal{R_C}<1.$
\end{center}
This gives necessary vaccination coverage to achieve stability of the disease free equilibrium.  However, it is more often the case the vaccine coverage is heterogeneous, therefore it is also useful to solve on vaccination proportion as a function of the other.  Given:
$$\mathcal{R_C} =\displaystyle \frac{\beta \phi \left((1-\delta_1)+(1-\delta_2) + \sqrt{4 \alpha^2 (1-\delta_1)(1-\delta_2) + (\delta_1 - \delta_2)^2}\right)}{2(\mu+\phi)(\mu+\gamma)},$$
we set:
$$\frac{\beta \phi \left((1-\delta_1)+(1-\delta_2) + \sqrt{4 \alpha^2 (1-\delta_1)(1-\delta_2) + (\delta_1 - \delta_2)^2}\right)}{2(\mu+\phi)(\mu+\gamma)}<1.$$
Now solve for some $D_i$ with the simplified form of $\mathcal{R_C}:$
\begin{eqnarray*}
	\mathcal{R_C}=\frac{\beta\phi\left(D_1+D_2+\sqrt{(D_2-D_1)^2+4\alpha^2D_1D_2}\right)}{2AB},\\
	\frac{\beta\phi\left(D_1+D_2+\sqrt{(D_2-D_1)^2+4\alpha^2D_1D_2}\right)}{2AB}<1,\\
	\beta\phi(D_1+D_2)+\beta\phi\sqrt{(D_2-D_1)^2+4\alpha^2D_1D_2}<2AB,\\
	\beta\phi\sqrt{(D_2-D_1)^2+4\alpha^2D_1D_2}<2AB-\beta\phi(D_1+D_2).\\
\end{eqnarray*}
It was previously shown that if $D_j\beta\phi>AB$ then $\mathcal{R_C}>1,$ therefore we know that the right side of the inequality must be positive given the assumption $\mathcal{R_C}<1.$  Continuing we have:
\begin{eqnarray*}
	(D_2-D_1)^2+4\alpha^2D_1D_2&<&\left(\frac{2AB}{\beta\phi}-(D_1+D_2)\right)^2,\\
	(D_2-D_1)^2+4\alpha^2D_1D_2&<&\left(\frac{2AB}{\beta\phi}\right)^2-\left(\frac{2AB}{\beta\phi}\right)(D_1+D_2)+(D_1+D_2)^2,\\
	(D_2-D_1)^2+4\alpha^2D_1D_2&<&\left(\frac{2AB}{\beta\phi}\right)^2-\left(\frac{2AB}{\beta\phi}\right)(D_1+D_2)+(D_1-D_2)^2 + 4D_1D_2,\\
	4\alpha^2D_1D_2-4D_1D_2&<&\left(\frac{2AB}{\beta\phi}\right)^2-\left(\frac{2AB}{\beta\phi}\right)D_1-\left(\frac{2AB}{\beta\phi}\right)D_2.\\
\end{eqnarray*}
At this point it becomes more useful to solve for some $D_i$ as solving $D_1,~D_2$ gives symmetric expressions in terms of the other.  Continuing gives:
\begin{eqnarray*}
	4\alpha^2D_iD_j-4D_iD_j&<&\left(\frac{2AB}{\beta\phi}\right)^2-\left(\frac{2AB}{\beta\phi}\right)D_i-\left(\frac{2AB}{\beta\phi}\right)D_j,\\
	D_i\left(4\alpha^2D_j-4D_j+\left(\frac{2AB}{\beta\phi}\right)\right)&<&\left(\frac{2AB}{\beta\phi}\right)^2-\left(\frac{2AB}{\beta\phi}\right)D_j,\\
	D_i&<&\frac{\left(\frac{2AB}{\beta\phi}\right)^2-\left(\frac{2AB}{\beta\phi}\right)D_j}{\left(4\alpha^2D_j-4D_j+\left(\frac{2AB}{\beta\phi}\right)\right)},\\
	D_i&<&\frac{AB(2AB-D_j\beta\phi)}{\beta\phi(AB+2\beta\phi D_j(\alpha^2-1))}.\\
\end{eqnarray*}
Substituting $D_i=1-\delta_i$ and $D_j=1-\delta_j$ allows us to solve for $\delta_i$:
$$\delta_i>1-\frac{AB(2AB-\beta\phi(1-\delta_j))}{\beta\phi(AB+2\beta\phi(\alpha^2-1)(1-\delta_j))}$$

As a result we have now derived the condition on $\delta_i$ to achieve $\mathcal{R_C}<1$ with $\delta_j$ and the other parameters fixed. So long as this condition holds, the DFE is stable.

\section{Sensitivity Analysis}

We want to study fixed point behaviour, i.e. disease-free and endemic equilibria with respect to parameter variations and $\mathcal{R}_0$. Considering that we will be working with time independent expressions, a forward sensitivity analysis will be sufficient.
The system of equations that represent the equilibria is shown below.

\begin{equation}
\centering
\begin{split}
- \beta s^*_1 (i^*_1 + \alpha i^*_2) - \mu s^*_1 + (1 - \delta_1)\mu &= 0 \\
 \beta s^*_1 (i^*_1 + \alpha i^*_2) - (\mu + \phi) e^*_1 &= 0 \\
 \phi e^*_1 - (\mu + \gamma)i^*_1 &= 0 \\
 - \beta s^*_2 (i^*_2 + \alpha i^*_1) - \mu s^*_2 + (1 - \delta_2)\mu &= 0 \\
 \beta s^*_2 (i^*_2 + \alpha i^*_1) - (\mu + \phi) e^*_2 &= 0 \\
 \phi e^*_2 - (\mu + \gamma)i^*_2 &= 0 \\
\end{split}
\label{eq:FixedPoints}
\end{equation}

The forward sensitivity problem is defined by

\begin{equation}
D_{\mathbf{u}^*} \cdot \frac{\partial \mathbf{u}^*}{\partial p} = - \nabla_p F
\label{eq:ForSP} 
\end{equation} 

Which can be solved multiplying both sides by $D_{\mathbf{u}^*}^{-1}$, given a nice enough Jacobian, thus giving

\begin{equation}
\frac{\partial \mathbf{u}^*}{\partial p} = - D_{\mathbf{u}^*}^{-1} \cdot \nabla_p F
\label{eq:ForSPsolved}
\end{equation}

Where,

\begin{equation}
D_{\mathbf{u}^*} = 
\begin{pmatrix}
- \beta (i_1^* + \alpha i_2^*) - \mu & 0 & -\beta s_1^* & 0 & 0 & -\beta \alpha s_1^* \\
\beta (i_1^* + \alpha i_2^*) & - (\mu + \phi) & \beta s_1^* & 0 & 0 & \beta \alpha s_1^* \\
0 & \phi & -(\mu + \gamma) & 0 & 0 & 0 \\
0 & 0 & -\beta \alpha s_2^* & -\beta (i_1^* \alpha + i_2^*) - \mu & 0 & -\beta s_2^* \\
0 & 0 & \beta \alpha s_2^* & \beta (i_1^* \alpha + i_2^*) & -(\mu + \phi) & \beta s_2^* \\
0 & 0 & 0 & 0 & \phi & -(\mu + \gamma) \\
\end{pmatrix}
\label{eq:Jacobian}
\end{equation}

\begin{equation}
\frac{\partial \mathbf{u}^*}{\partial p} =
\left(\frac{\partial s_1^*}{\partial p} \:,\: \frac{\partial e_1^*}{\partial p} \:,\: \frac{\partial i_1^*}{\partial p} \:,\: \frac{\partial s_2^*}{\partial p} \:,\: \frac{\partial e_2^*}{\partial p}  \:,\: \frac{\partial i_2^*}{\partial p}\right)^T \\
\label{eq:PartialDerVector}
\end{equation}

\begin{equation}
\begin{split}
\nabla_{\beta}F &=
\left(-s_1^*(i_1^* + \alpha i_2^*) \:,\: s_1^*(i_1^* + \alpha i_2^*) \:,\: 0 \:,\: -s_2^*(i_2^* + \alpha i_1^*) \:,\: s_2^*(i_2^* + \alpha i_1^*)  \:,\: 0 \right)^T \\
\nabla_{\alpha}F &=
\left(-s_1^* \alpha i_2^* \:,\: s_1^* \alpha i_2^* \:,\: 0 \:,\: -s_2^* \alpha i_1^* \:,\: s_2^* \alpha i_1^*  \:,\: 0 \right)^T \\
\nabla_{\mu}F &=
\left(-s_1^* + 1 - \delta_1 \:,\: -e_1^* \:,\: -i_1^* \:,\: -s_2^* + 1 - \delta_2 \:,\: -e_2^* \:,\: -i_2^*  \right)^T \\
\nabla_{\phi}F &=
\left( 0 \:,\: -e_1^* \:,\: e_1^* \:,\: 0 \:,\: -e_2^* \:,\: e_2^*  \right)^T \\
\nabla_{\gamma}F &=
\left( 0 \:,\: 0 \:,\: -i_1^* \:,\: 0 \:,\: 0 \:,\: -i_2^*  \right)^T \\
\nabla_{\delta_1}F &=
\left( -\mu \:,\: 0 \:,\: 0 \:,\: 0 \:,\: 0 \:,\: 0 \right)^T \\
\nabla_{\delta_2}F &=
\left( 0 \:,\: 0 \:,\: 0 \:,\: -\mu \:,\: 0 \:,\: 0 \right)^T \\
\end{split}
\label{eq:SetofPar_Derivatives}
\end{equation}

At the DFE,

\begin{equation}
D_{\mathbf{u}^*} = 
\begin{pmatrix}
- \beta - \mu & 0 & -\beta (1 - \delta_1) & 0 & 0 & -\beta \alpha (1 - \delta_1) \\
0 & - (\mu + \phi) & \beta (1 - \delta_1) & 0 & 0 & \beta \alpha (1 - \delta_1) \\
0 & \phi & -(\mu + \gamma) & 0 & 0 & 0 \\
0 & 0 & -\beta \alpha (1 - \delta_2) & -\beta - \mu & 0 & -\beta (1 - \delta_2) \\
0 & 0 & \beta \alpha (1 - \delta_2) & 0 & -(\mu + \phi) & \beta (1 - \delta_2) \\
0 & 0 & 0 & 0 & \phi & -(\mu + \gamma) \\
\end{pmatrix}
\label{eq:JacatDFE}
\end{equation}

The matrix shown above correspond to the Jacobian of the system (\ref{eq:FixedPoints}) at the DFE. Its inverse, which in this case can be computed without having a considerable error, substituted on expression (\ref{eq:ForSPsolved}), give the vector of partial derivatives of the states respect to the parameters.      

\subsection{Sensitivity Indices of DFE}

The expressions below are obtained by SI formula considering the solution of the FSP (\ref{eq:ForSPsolved}) for the DFE:

\begin{equation}
\begin{split}
S_{\delta_1} &= \frac{\delta_1}{\delta_1 - 1}, \\
S_{\delta_2} &= \frac{\delta_2}{\delta_2 - 1}.\\
\end{split}
\label{eq:Sens_for_DFE}
\end{equation}

These equations represent the effect of a change in vaccination coverage on the DFE system state. Which means that if we increase or decrease the vaccination coverage $\delta_i$ by $1 \%$, then the state value $s^*_i$ will be modified by a factor of $S_{\delta_i}$.
To represent this numerically we substitute $\delta_1 = 0.85 \:,\: \delta_2 = 0.87$, (estimated from \cite{francefrancefrancelol}), into the above expressions, giving:

\begin{equation}
\begin{split}
S_{\delta_1} &= -5.7, \\
S_{\delta_2} &= -6.7. \\
\end{split}
\label{eq:Sens_for_DFE:NV}
\end{equation}

If we analyze the relation between the modulus of each sensitivity index, we can conclude that this state is slightly more sensitive to $\delta_2$-variations. $|S_{\delta_2}|$ is $17\%$ greater than $|S_{\delta_1}|$. This result tells us that a variation of $1\%$ on the value of $\delta_2$ will affect the DFE system state $1.17$ times more than a $1\%$ variation of $\delta_1$. Thus the DFE is more sensitive to $\delta_2$ perturbations. 

\subsection{Sensitivities Indices on EE}

The endemic equilibrium state is obtained solving system (\ref{eq:FixedPoints}), for the condition $i_1^* \:,\: i_2^* \: \neq 0$. The solution is composed of two states, neglecting complex conjugates, where only one of them has biological meaning. Explicit forms of the sensitivity indices for the endemic will not be included, due to the length and complexity of the algebra. In future work, a numerical simulation for this result would endow the analysis with more meaningful interpretation.

\subsection{Sensitivity for $\mathcal{R}_0$}

We consider the importance of an accurate estimation of $\mathcal{R}_0$ (threshold for diseases outbreaks). We have proceeded to find the SI (Sensitivity Indices) for the basic reproductive number and have studied the dependence on the parameters.
The general results are shown below.

\begin{equation}
\centering
\scalemath{0.79}{
\begin{split}
S_{\phi} &= \frac{\mu}{\mu + \phi} \qquad
S_{\beta} = 1 \qquad
S_{\gamma} = -\frac{\gamma}{\mu + \gamma} \qquad
S_{\mu} = -\frac{\mu(2\mu + \phi + \gamma)}{(\gamma + \mu)(\mu + \phi)} \\
S_{\delta_1} &= \frac{\delta_1 \left[-1 + \frac{-2 \alpha^2 (1 - \delta_2) + \delta_1 - \delta_2}{\sqrt{4\alpha^2 (1-\delta_1)(1-\delta_2) + (\delta_1 - \delta_2)^2}}\right]}{2 - \delta_1 - \delta_2 + \sqrt{4\alpha^2(1- \delta_1)(1 - \delta_2) + (\delta_1 - \delta_2)^2}} \\
S_{\delta_2} &= -\frac{\delta_2 \left[1 + \frac{-2 \alpha^2 (1 - \delta_1) + \delta_1 - \delta_2}{\sqrt{4\alpha^2 (1-\delta_1)(1-\delta_2) + (\delta_1 - \delta_2)^2}}\right]}{2 - \delta_1 - \delta_2 + \sqrt{4\alpha^2(1- \delta_1)(1 - \delta_2) + (\delta_1 - \delta_2)^2}} \\
S_{\alpha} &= \frac{4 \alpha^2 (1 - \delta_1)(1 - \delta_2)}{\sqrt{4 \alpha^2 (1 - \delta_1)(1 - \delta_2) + (\delta_1 - \delta_2)^2}\left(2 - \delta_1 - \delta_2 + \sqrt{4 \alpha^2 (1 - \delta_1)(1 - \delta_2) + (\delta_1 - \delta_2)^2}\right)} \\
\end{split}
}
\label{eq:Sens_for_R_o}
\end{equation}

Observing the SI expressions (\ref{eq:Sens_for_R_o}) we can state that a strong dependence is focused on $\delta_1$, $\delta_2$, and $\alpha$. A variation on these and other parameters will converge in a potential misestimation of $\mathcal{R}_0$, depending on the perturbation direction. Because we estimated $\beta$ using $\mathcal{R}_0$, $\beta$ is in fact dependent on all other parameter values. For this reason, although the $\beta$-index appears to be constant in the above expression, in truth it has considerable dependence on the other parameters.

\section{Numerical Simulations}
\subsection{Two-Patch Simulation} %scaled preferential mixing
For our two-patch numerical simulation, we assume the population in the Northern and Southern regions are equal, and we scale by 0.35 (which was estimated from CIA data, \cite{france_factbook}) to reflect that our model is concerned primarily with adolescents. We use the following numerical values for the parameters (See Appendix A for details on estimations): $\alpha = 0.02, \beta = 2.2, \gamma = 0.25, \phi = 0.125, \delta_1 = 0.88, \delta_2 = 0.85, \mu = 3\times10^{-5}.$\\

\begin{figure}[!htbp]
\centering
\includegraphics[width=\textwidth]{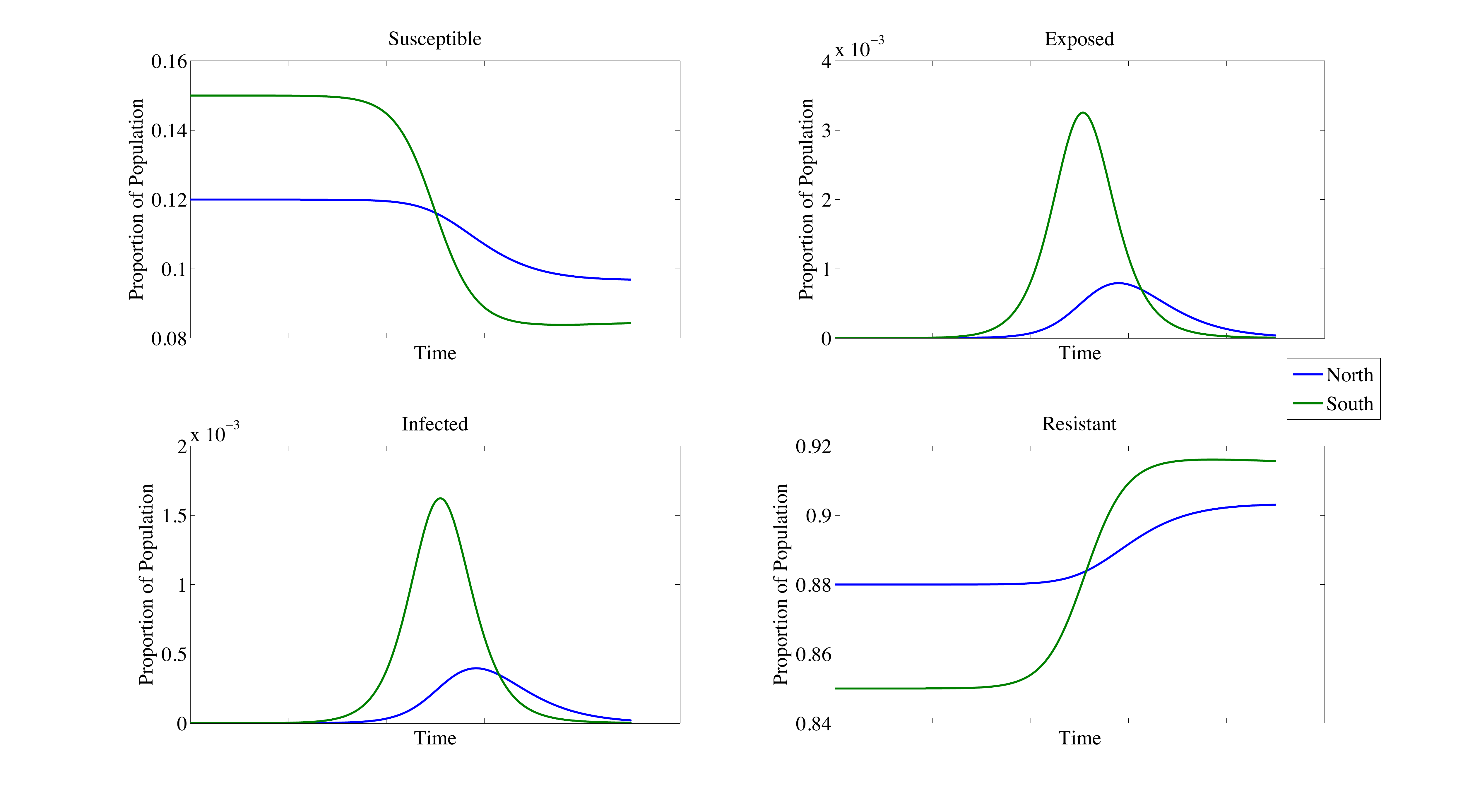}
\caption{Two Patches: Scaled Preferential Mixing \textit{(MATLAB Simulation)}}
\end{figure}

For reference, we also run simulations in which $\delta_1 = \delta_2 = 0.9 \text{ and } \alpha = 1.$ In these simulations, the parameters in both patches are identical, and $\alpha = 1$ means that individuals interact across patches as often as they interact within their home patch. We see that the system behaves as one large patch, which is as expected. In addition, one can observe that the final epidemic size is consistently higher in the homogenous system: 1.5 million infected in the homogeneous system as opposed to 1 million in the heterogeneous two-patch system.\\

\begin{figure}[!htbp]
\centering
\includegraphics[width=\textwidth]{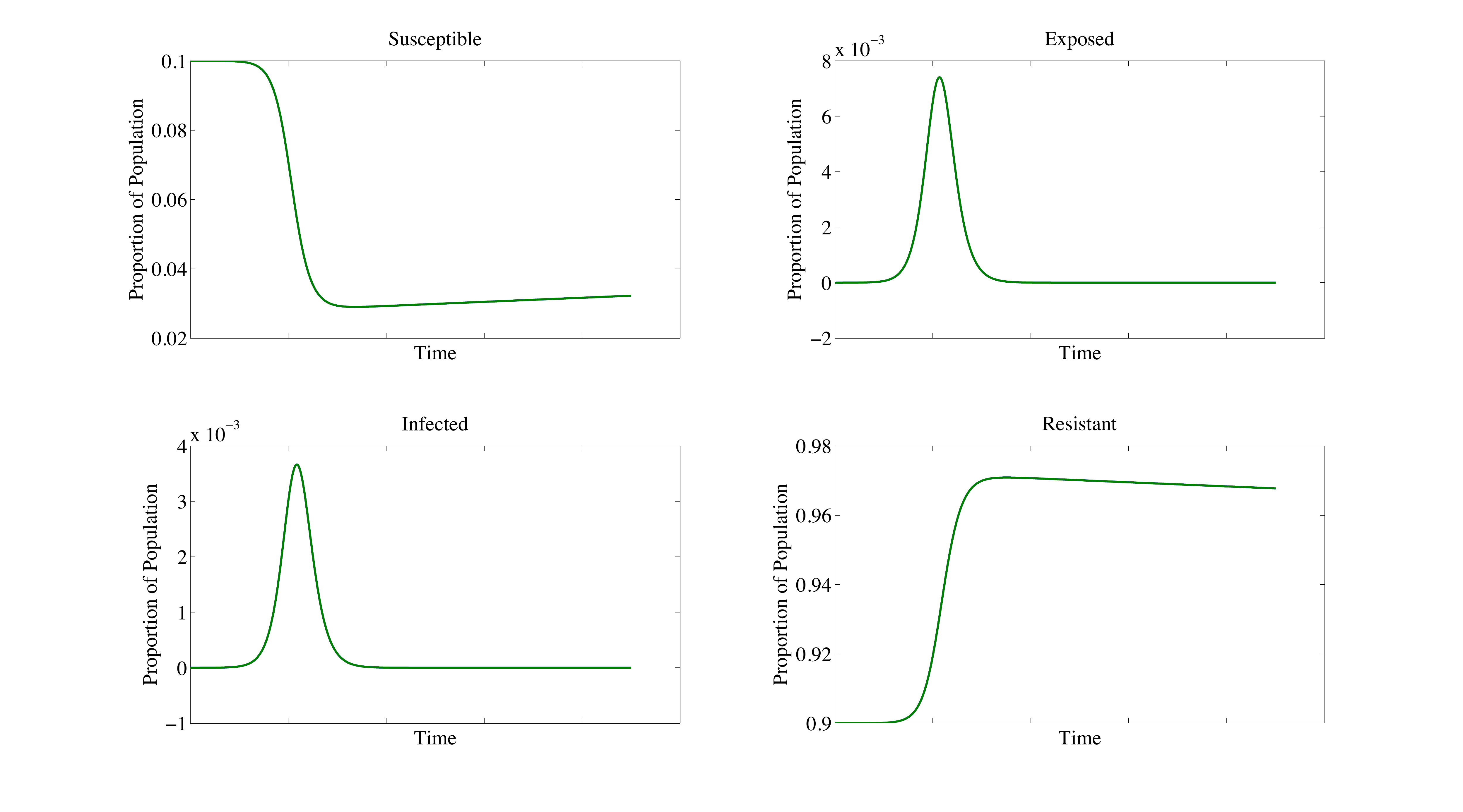}
\caption{Two Patches: parameters set to represent homogeneity \textit{(MATLAB Simulation)}}
\end{figure}

From the two-patch simulations it is apparent that there are substantial differences in the qualitative dynamics of the disease outbreak between the heterogeneous and homogeneous cases.
\subsection{Multi-Patch Simulation} %preferential mixing, break to keep 2patch figure with above text
For our numerical simulation of multiple patches, we divide France into six metropolitan regions, based on regional demographics. We shall use the following parameter estimates: \\

\begin{table}[htp]
\centering
\begin{tabular}{rccc}
	Patch&$N\times10^6$&$\delta$&$\pi$\\
	\hline
	Strasbourg & 5.82 & 0.855 & 0.9715\\
	Paris & 6.13 & 0.874 & 0.9370 \\
	Toulouse & 3.58 & 0.779 & 0.9628 \\
	Nantes & 2.54 & 0.836 & 0.9602 \\
	Marseille & 1.73 & 0.798 & 0.9613\\
	Lyon&2.17&0.855&0.9577\\
\end{tabular}
\caption{Parameters by Region}
\end{table}

%\begin{figure}[!htbp]
%\centering
%\includegraphics[scale=0.75]{HippieFrenchDensity.png}
%\caption{put this in appendix ?}
%\end{figure}

\begin{figure}[!htbp] 
\centering
\includegraphics[width=\textwidth]{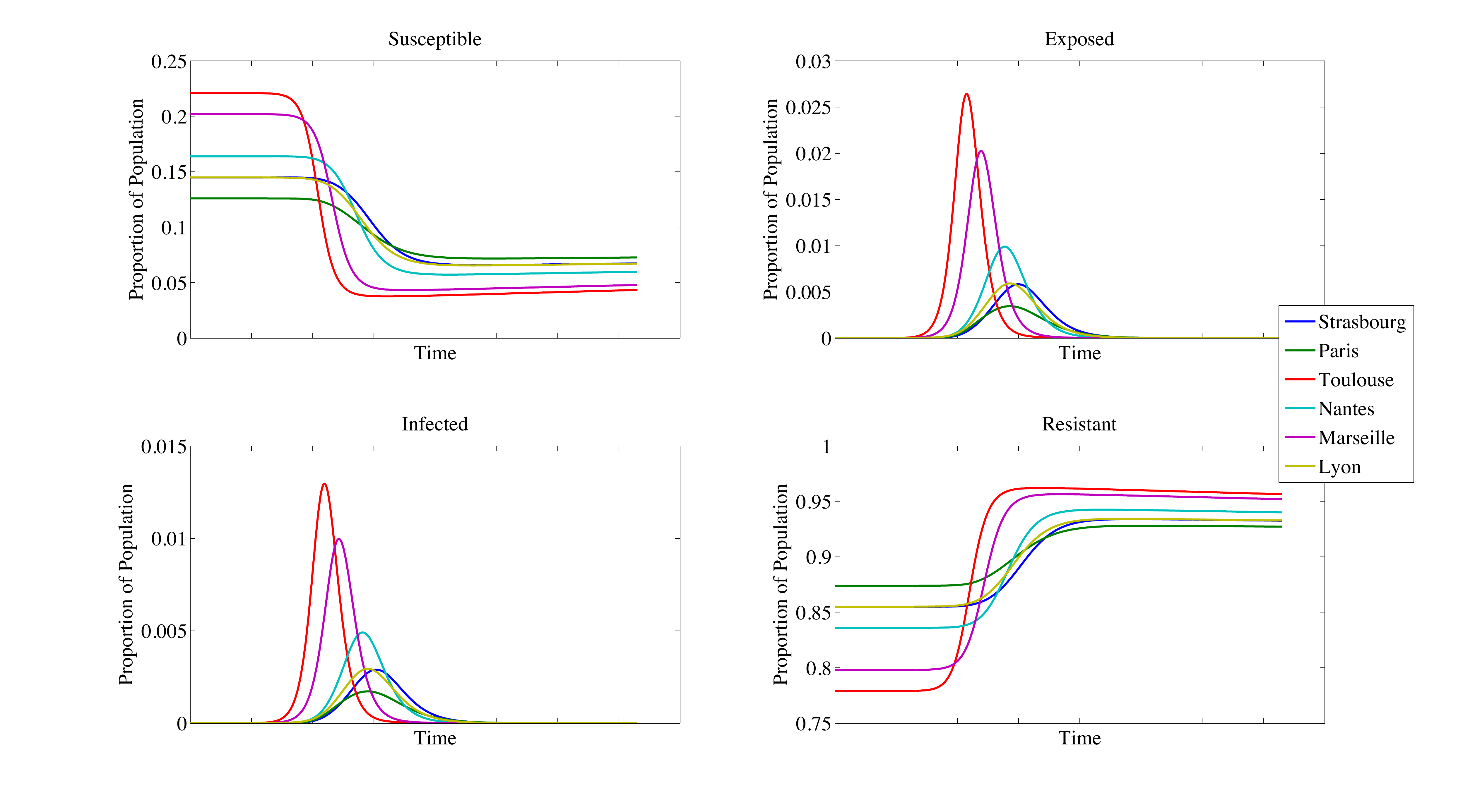}
\caption{Six Patches: Preferential Mixing \textit{(MATLAB Simulation)}}
\end{figure}

Once again, we run simulations with six patches to represent homogeneity. We set all varying parameters to the average: $\delta_i = 0.833, \pi_i = 0, N = 3.6$ million. One can see that the six patches behave as one large patch, and the final epidemic size is larger than with separate heterogeneous patches. \\

\begin{figure}[!htbp] 
\centering
\includegraphics[width=\textwidth]{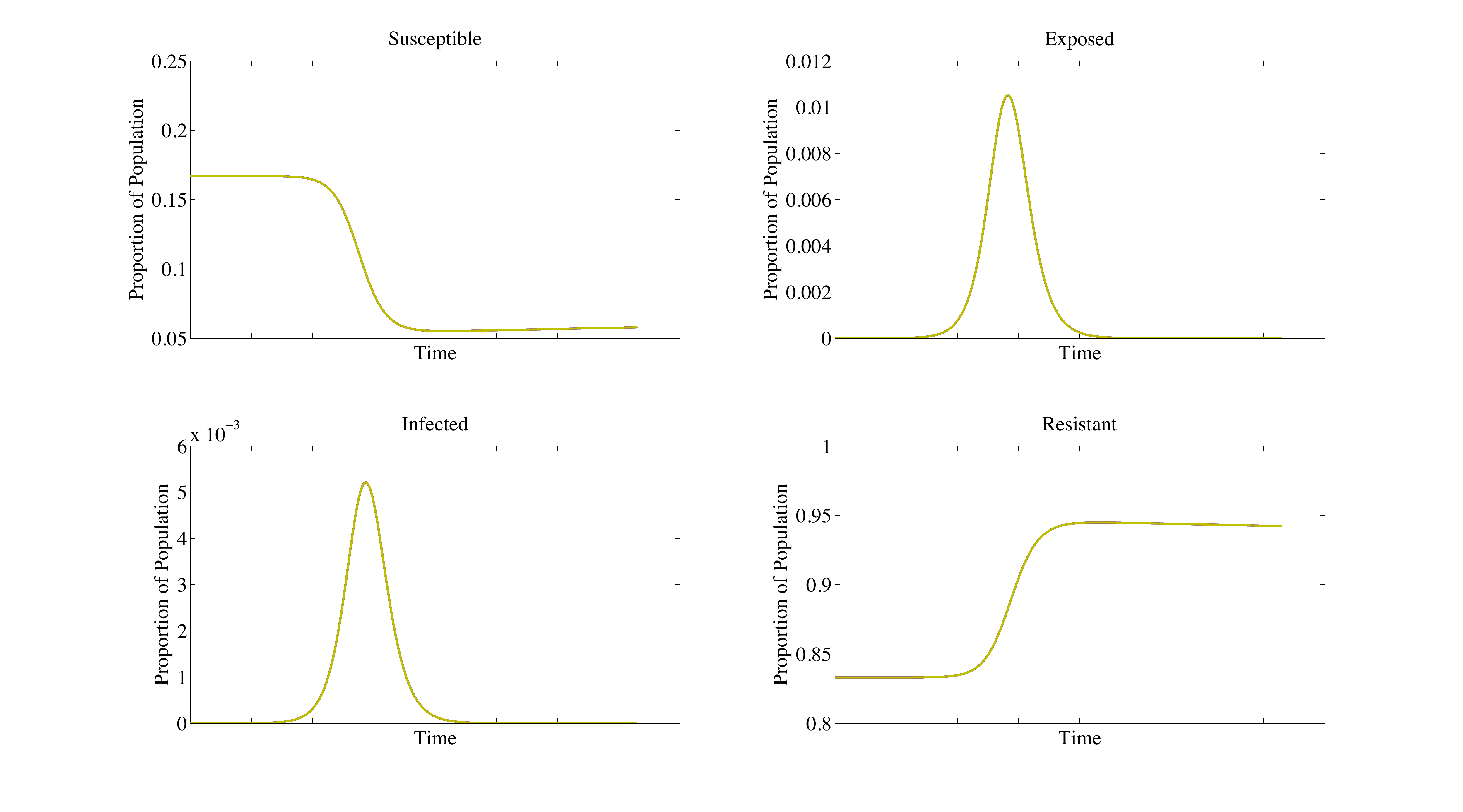}
\caption{Six Patches: parameters set to represent homogeneity \textit{(MATLAB Simulation)}}
\end{figure}

The six-patch simulation agrees with the two-patch in that qualitative dynamics are substantially different between the heterogeneous and homogeneous cases.\\
\break \section{Results and Discussion}
\subsection{Results}
%Numerical simulations show that for the realistic parameter values calculated, likelihood of an outbreak is quite high, however the likelihood of an endemic equilibrium is low.  Mathematical analysis, particularly of the disease free equilibrium, shows that if the $\mathcal{R_C}<1$ the disease will die off and the DFE is stable.  This indicates that the $\delta$ parameters are essential in controlling the reproduction of the disease.  This is obvious in the suggested manage of measles by public health officials.  The heterogeneity of vaccination coverage adds an interesting element in that regional coverage must increase above recommended coverage if another nearby by area is substantially below.  Another interesting element discovered in the heterogeneity of interaction is that mixing between patches plays an incredibly important role in the dynamics of the disease and indeed has direct correlation to the final epidemic size.  When these two heterogeneous elements combine it becomes clear that effective intervention and public health strategies must take into account such complexities.  Simple national policies and recommendations are insufficient to guarantee protection from measles, especially if adherence to such policies is incomplete.  Regional consideration could easily be advised based on analysis of circumstance and projections such as those contained here and in similar literature.

In the analysis of $\mathcal{R_C}$ for the two patch case, it was demonstrated that heterogeneous vaccine coverage increases $\mathcal{R_C}$ in relation to comparable homogeneous coverage. Because the disease-free equilibrium is stable given $\mathcal{R_C}<1,$ one can see that heterogeneous coverage makes a stable disease-free equilibrium less likely. Functionally this means that in addition to vaccination coverage within each region, $\mathcal{R_C}$ depends on the difference in vaccination coverage across regions. Expressions have also been determined for when $\delta_1=\delta_2$ and $\delta_1 \neq \delta_2$ that allow the computation of necessary vaccination coverages to drive $\mathcal{R_C}<1.$  This essentially defines the necessary levels to achieve herd immunity of the population as a whole in the model.\\
\\
For the endemic case, the method of substitution was used to solve for the endemic equilibrium $(s_1^*,e_1^*,i_1^*,s_2^*,e_2^*,i_2^*)$ implicitly, such that all fixed points were expressions in
$e_1^*$ and $e_2^*.$ It was then determined that $e_i^* < \frac{\mu (1-\delta_i)}{\mu+\phi},$ which follows directly from the requirement that the endemic equilibrium be positive and finite. This condition is the greatest upper bound, because as $e_i^*$ approaches this bound, $e_j^* \rightarrow \infty.$  Numerical methods were also used to evaluate the existence and stability of the endemic equilibria.\\
\\
From a forward sensitivity analysis of $\mathcal{R_C}$ in our scaled preferential mixing model, we found that $\mathcal{R_C}$ is most sensitive to perturbations in $\delta_1$ and $\delta_2.$ Using our parameter estimates, (see Section 6. Numerical Simulations), we find that the sensitivity indices themselves are strongly influenced by the $\delta$ parameters, and $\mathcal{R_C}$ is inversely sensitive to both $\delta_1$ and $\delta_2.$ The sensitivity indices $S_{\delta_1}$ and $S_{\delta_2}$ will be large and negative when $\delta_1$ and $\delta_2$ are large and similar, and $S_{\delta_1}$ and $S_{\delta_2}$ will be small and negative when $\delta_1$ and $\delta_2$ are small and similar. For discrepancies between $\delta_1$ and $\delta_2,$ we generally find that $S_{\delta_i}$ is large and negative while $S_{\delta_j}$ is small and negative for $\delta_i < \delta_j.$ So $\mathcal{R_C}$ is largely inversely sensitive to the lower vaccination rate, see Table 1.

\begin{table}[h!]
\centering
\begin{tabular}{cccc|l}
$\delta_1$&$\delta_2$&$S_{\delta_1}$&$S_{\delta_2}$& \\
\hline
0.88 & 0.85 & -0.05 & -5.60 & our estimated values\\
0.90 & 0.90 & -4.50 & -4.32 & two high values\\
0.8. & 0.8. & -2.00 & -1.92 & two low values\\
0.87 & 0.86 & -0.40 & -5.72 & small discrepancy, $\delta_1>\delta_2$\\
0.86 & 0.87 & -5.78 & -0.34 & small discrepancy, $\delta_1<\delta_2$\\
0.90 & 0.83 & -0.01 & -4.87 & large discrepancy, $\delta_1>\delta_2$\\
0.83 & 0.9. & -4.88 & -0.00 & large discrepancy, $\delta_1<\delta_2$\\
\end{tabular}
\caption{Variations of $\delta_1$ and $\delta_2$}
\end{table}

\subsection{Discussion}
Our findings suggest that heterogeneous MMR vaccination coverage within France does in fact increase the transmissibility of measles and contribute to the likelihood of outbreaks. Our model showed that $\mathcal{R_C}$ increased to a larger value when the disparity in vaccination coverage between regions widened, even when the average vaccination coverage between the two patches remained the same. It thus follows that health policy officials should first focus their immunization efforts primarily on districts with the lowest vaccination coverage rates when attempting to eradicate the disease at a regional level. This will weaken the impact that the lowest covered districts have on moderately covered regions and will also diminish the diseases communicability throughout surrounding populations. \\
\\
In addition, it is worth noting that while narrowing the gap between vaccination coverage percentages across regions will help decrease virus transmission, this should not be the sole strategy implemented when considering national vaccination coverage. The optimal average immunization coverage should still be achieved in order to eliminate the possibility of another measles epidemic. 

\subsection{Future Research}
In this project many biological aspects of the model were neglected for the sake of simplicity. For example, an age-structured model was not considered in this report, instead it was assumed that all individuals were equally susceptible to the disease regardless of age. However, this type of model would have represented the disease dynamics in a more accurate light since children under 1 year of age are more likely to contract the measles virus than any other age group \cite{8586668720130301}. Additionally, the epidemic in France from 2008-2011 showed a dramatic increase in the rate of young adults contracting the disease; the median age during the third outbreak was 16 years \cite{8586668720130301}. In addition, vaccine efficacy was considered at the outset by taking the product of estimated coverages and efficacy as the effective vaccine rate rather than having efficacy degrade over time as in \cite{gumel}.  Both of these modeling decision are likely to have a relevant impact on the dynamics of the disease and would be important considerations in the future.

\pagebreak
\begin{appendices}
\section{Parameter Calculation} 
\subsection{$N_i$ - population size}
The overall population for France is easily found from national demographics \cite{france_factbook}.  In mathematical analysis and numerical simulation of the two-patch model, the population sizes were assumed equal.  Therefore we simply scaled the total population to achieve an effective population given the significantly higher prevalence in younger age groups and divided it in half.  For the multi-patch model, scaling was again done to achieve an effective population and regional demographic data was used to divide the population among patches.
\subsection{$\mu$ - birth and death rate}
For the sake of simplicity, it was assumed that the population constant.  The birth and death rate were then necessarily equal.  It was decided to use an estimation of the birth rate to represent $\mu.$  (France has a birth rate $>$ death rate which motivated this decision \cite{france_factbook}).  The dynamics of a non-constant population are different than that of a constant, but in having the demographic modelling rate be higher a closer approximation can be reached.  From the World Factbook maintained by the CIA we found that France has $\approx12.5~\mbox{births}/1000$ residents annually \cite{france_factbook}.  This implies that $\mu=\frac{12.5/1000}{365}\approx0.00003.$
\subsection{$\phi$ - disease incubation rate}
The duration of the exposed or incubation period is vital in modelling disease with such latent periods. With any disease though, exact time frames of all stages of the disease vary.  The CDC lists the incubation period to be between one and three weeks \cite{CDC_webpage}.  However they also state that the contagious period begins 4 days prior to the characteristic rash with the rash generally appearing at 14 day post exposure \cite{CDC_webpage}.  This would suggest the latent period to be between 7 and 10 days.  Research literature that focuses on modelling measles \cite{hooker2010parameterizing, keeling2002understanding} suggests that 8 days is optimal for replicating disease dynamics.  Using these figures one can estimate $\phi$ by considering $\frac{1}{\phi}=\mbox{average latent period}.$  This approximates the range $\phi\approx[0.1,0.14].$
\subsection{$\gamma$ - recovery rate}
The infectious period is another critical parameters to determine.  Epidemiological data suggests the infectious period can range from 4 to 10 days, with the CDC stating the average is about 8 \cite{CDC_webpage}.  Again turning to modelling literature, the estimates are often on the low end of this range \cite{keeling2002understanding, hooker2010parameterizing}.  This is due to the partial isolation that occurs from highly acute measles infection.  Calculating $\gamma$ then proceeds as for $\phi$ with $\frac{1}{\gamma}=\mbox{average infective period}$ giving the range $\gamma\approx[0.17,0.25].$ 
\subsection{$\delta_i$ - vaccination coverage}
Vaccination coverage and disparity coverage is the driving motivation of this project, it is therefore important to ensure accurate estimations.  France's average vaccination coverages fluctuates around $90\%,$ however it becomes necessary to consider both regional coverage and vaccine efficacy in calculating $\delta_i$ accurately.  Literature provides a good basis for regional vaccination approximation \cite{8586668720130301} and the only thing to then consider is vaccine efficacy.  It is know that any vaccination program fails to be $100\%$ effective, this can be represented differently.  In our model the simplifying assumption was made that if the vaccine proves to be effective ones transitions from natural immunity to vaccinated immunity with no interruption of protection, thus people are "born" into either susceptible or resistant classes.  To continue maintaining the simplicity of the model, we decided to multiply the vaccination coverage by the estimated vaccine efficacy and use this as $\delta_i.$  Literature estimates of efficacy vary widely \cite{bouhour2012survey, akramuzzaman2002measles}, at times spanning from $80\%$ to as high as $98\%.$  The lower end of the range generally occurs in third world countries where quality of health care is lower in in particular cases when only a single dose of vaccine is administered and is given too early to be fully effective.  The high end of the range is in the case where two MMR doses are administered as recommended.  Most of the literature focusing on efficacy in Europe \cite{bouhour2012survey, mossong2000estimation, eichner2002estimation} suggest an accurate range to fall between $92\%$ and $96\%$ nationwide.  We opted to consider efficacy as $95\%$ with coverage between $84\%$ and $94\%.$  This gives a reasonable range as $\delta=\mbox{coverage}*\mbox{efficacy}=[0.798,0.893].$
\subsection{$\pi$ - inter-patch mixing proportions}
A necessary parameter to consider for the preferential mixing is $\pi.$  This parameter represents the proportion of individuals from a given patch that mix with the entire population, and was a difficult parameter to estimate.  The people-days concept became important in its computation.  Colloquially we state that one people-day is the time, in days, spent by a person in some place.  Therefore it can be considered a measure of activity.  It becomes spatially relevant when considering if a person is on their home patch or traveling.  Our application of this concept is defined by considering the mean travel time and number of travelers.  This allows the calculation of the proportion of people-days that are spent elsewhere by the inhabitants of a given patch.  While the model does not explicitly represent this movement, we use it as an estimation of the proportion of a given patch that is mixing with other patches, namely:

\begin{equation}
1  \, \mbox{people-day} = \: \mbox{person} * \mbox{day}.
\end{equation}      

Once defined thus, we are able to calculate the proportion of travellers in a patch.  In terms of people-days, this becomes equitable to the amount of activity or contact a patch can have with other patches, as follows:

\begin{equation}
\mbox{Proportion of patch travelling} = \frac{\# \mbox{Travellers} \: * \: \mbox{Mean Travel Time}}{\mbox{Total Patch Population} *\mbox{365 days}}.
\end{equation}

This proportion of patch travelling became the foundation for the $\pi$ parameter.  It was adjusted slightly to account for the lower likelihood of travel and thus between patch interaction of persons in the birth - 20 year old age range which comprised the majority of measles cases.  Our final range in the two-patch was $\pi\approx[0.96,0.98]$ and in the multi-patch system was $\pi\approx[0.94,0.99].$  The following table shows the necessary data for completing such calculations, it was found in \cite{Memento}.

\begin{table}[h!]
\scalebox{0.7}{
\begin{tabular}{l|l|c|c|c}
Regions of France & Capital & \# Travellers [millions] & Mean Travel Time [Days] & \# Population [millions] \\
\hline
Alsace & Strasbourg & 3861 & 5 & 1856 \\
Aquitaine & Bordeaux & 8798 & 4.5 & 3227 \\
Auvergne	 & Clermont-Ferran & 3738 &	4.4 & 1345 \\ 
Bourgogne & Dijon & 4680 & 4.8 & 1643 \\
Bretagne	 & Rennes & 7820 & 4.8 & 3195 \\
Centre & Orleans & 7118 & 4.7 & 2545 \\
Champagne-Ardenne & & 3352 & 5.2 & 1334 \\	
Corse & Ajaccio & 624 & 5 & 311 \\
Franche-Comt\'e & Besan\c{c}on & 2787 & 5.2 & 1173 \\		
\^Ile-de-France & Paris & 43039 & 5.1 & 11798 \\
Languedoc-Roussillon & Montpellier & 6392 & 6.3 & 2633 \\
Limousin & Limoges & 2223 & 4.9 & 746 \\		
Lorraine & Metz & 5110 & 4.3	 & 2350 \\
Midi-Pyr\'en\'ees & Toulouse & 8529 & 5.1 & 2893 \\ 
Nord-Pas-de-Calais & Lille & 6796 & 4.6 & 4026 \\	
Basse-Normandie & Caen & 3767 & 6.8 & 1474 \\
Haute-Normandie & Rouen & 3998 & 4.5 & 1833 \\
Pays de la Loire & Nantes & 10993 & 5.8 & 3565 \\
Picardie	 & Amiens & 4039 & 4.7 & 1914 \\
Poitou-Charentes & Poitiers & 5601 & 4.9 & 1774 \\ 	
Provence-Alpes-C\^ote d'Azur & Marseille & 12697 & 5.5	& 4951 \\
Rh\^one-Alpes & Lyon & 19967 & 4.8 & 6212 \\
\end{tabular}
}
\label{TAB:FranceTravellers}
\caption{\footnotesize{Breakdown of annual travel and associated travel time originating from the regions of France.  Note that number of travellers is not individuals but rather the sum of all records showing a traveller originating from the given region.}}
\end{table}

\subsection{$\beta$ - scaled successful contact rate}
Being a scaled parameter, $\beta$ does not have as precise a biological interpretation as normal contact or transmission rates.  However, it is not entirely incorrect to think of as successful contact rate rescaled to account for the particulars of preferential mixing.  Given its somewhat ambiguous nature however, it is challenging to estimate it \emph{a priori} from any data or literature.  We were able to deduce it from our expression for $\mathcal{R}_0$ however.  Literature says that $\mathcal{R}_0$ can vary from 6 to 45 for measles.  However most previous modelling suggests 10 to 20 being a reasonable range.  Starting from our $\mathcal{R}$ we assume no vaccination and arrive at $\mathcal{R}_0,$ as follows:
\begin{center}
$$\mathcal{R}_0 = \displaystyle \frac{\beta \phi (1+\alpha)}{(\mu+\phi)(\mu+\gamma)}.$$
\end{center}
Solving for $\beta$ one obtains
\begin{center}
$$\beta= \frac{\mathcal{R}_0 (\mu+\phi)(\mu+\gamma)}{\phi (1+\alpha)}.$$
\end{center}
Substitute in $\mathcal{R}_0=10$ and $\mathcal{R}_0=20$ gives the range $\beta\approx[1.1,2.2].$

\subsection{$\alpha$ - scaled mixing parameter}
In the process of rescaling the two-patch preferential mixing model a new parameter $\alpha$ was generated.  This parameter can be found analytically using the previously determined values for $\pi$ and the definitions of $\alpha$ and the intermediate rescaling parameter $\rho.$  Given the $\rho=1-\pi$ and $\alpha=\frac{\rho}{1-\rho},$ a reasonable paramter range can be calculated to $\alpha\approx[0.01,0.02].$

\end{appendices}
\pagebreak
\section*{Acknowledgments} 
We would like to thank Dr.~Carlos Castillo-Chavez, Executive Director of the Mathematical and Theoretical Biology Institute (MTBI), for giving us the opportunity to participate in this research program.  We would also like to thank Co-Executive Summer Directors Dr.~Baojun Song and Dr.~Omayra Ortega for their efforts in planning and executing the day to day activities of MTBI. We also want to give special thanks to Abba Gumel(specify any names of research advisors, grad helpers, etc. that you feel have helped). This research was conducted in MTBI at the Simon A. Levin Mathematical, Computational and Modeling Sciences Center (SAL MCMSC) at Arizona State University (ASU). This project has been partially supported by grants from the National Science Foundation (NSF - Grant DMPS-1263374), the National Security Agency (NSA - Grant H98230-13-1-0261), the Office of the President of ASU, and the Office of the Provost of ASU.

\pagebreak

\end{document}